\documentclass[journal,comsoc]{IEEEtran}

\usepackage{graphicx}
\usepackage{bbm}
\usepackage{xcolor}
\usepackage{amsmath}
\usepackage{varwidth}
\usepackage{tikz}
\usepackage{pgf-umlsd}
\usetikzlibrary{shapes,arrows,shadows} 
\usepackage{bm}
\usepackage{bbm}
\usepackage{amssymb}
\usepackage{graphicx}
\usepackage{amsthm}
\usepackage[cmintegrals]{newtxmath}
\usepackage{algorithm}
\usepackage{algpseudocode}
\usepackage{pgf}
\usetikzlibrary{arrows,automata}
\usepackage{pgf-umlsd}
\usetikzlibrary{calc}
\usepackage[latin1]{inputenc}
\usetikzlibrary{decorations.pathreplacing}
\usetikzlibrary{shapes.geometric}
 \graphicspath{{images/}}

\DeclareMathOperator*{\expectation}{E}
\newcommand{\Exp}[1]{\expectation \left[ #1\right]}
\usepackage{subcaption}
\DeclareRobustCommand{\stirling}{\genfrac\{\}{0pt}{}}

\definecolor{TUMBlau}{RGB}{0,101,189} 
\definecolor{TUMBlauDunkel}{RGB}{0,82,147} 
\definecolor{TUMBlauHell}{RGB}{152,198,234} 
\definecolor{TUMBlauMittel}{RGB}{100,160,200} 

\definecolor{TUMElfenbein}{RGB}{218,215,203} 
\definecolor{TUMGruen}{RGB}{162,173,0} 
\definecolor{TUMOrange}{RGB}{227,114,34} 
\definecolor{TUMGrau}{gray}{0.6} 

\newtheorem{thm}{Theorem}[section]

\newcommand{\mgrev}[1]{\textcolor{black}{#1}}
\newcommand{\mgrevv}[1]{\textcolor{black}{#1}}
\newcommand{\mgrevvv}[1]{\textcolor{black}{#1}}

\usepackage{booktabs} 

\begin{document}
%

\title{ Admission Control based Traffic-Agnostic \\ Delay-Constrained Random Access  \\ (AC/DC-RA) for M2M Communication
	}


\author{H. Murat G\"ursu,~\IEEEmembership{Student Member,~IEEE,} Mikhail Vilgelm,~\IEEEmembership{Student Member,~IEEE,} Alberto M. Alba,~\IEEEmembership{Student Member,~IEEE,}, Matteo Berioli, Wolfgang Kellerer,~\IEEEmembership{Senior Member,~IEEE,}}
%

%


\maketitle 





\begin{abstract}
  The problem of wireless M2M communication is twofold: the reliability aspect and the scalability aspect. The solution of this problem demands a delay constrained random access protocol. To this end we propose Admission Control based Traffic Agnostic Delay Constrained Random Access (AC/DC-RA) protocol. Our main contribution is enabling the stochastic delay constraints agnostic to the traffic, such that the stochastic delay constraint is valid with respect to varying number of arrivals. We achieve this with an admission control decision that uses a novel collision estimation algorithm for active number of arrivals per contention resource. We use an adaptive contention resolution algorithm to react to the varying number of arrivals. Using these tools, the admission control solves the stability problem. We show with simulations that AC/DC-RA provide stochastic delay constrained random access in a traffic agnostic way to sustain a stable performance against Poisson and Beta arrivals without any modification to the protocol. 
 \end{abstract}

\section{Introduction}

In the last ten years, the yearly productivity in U.S. grew $42 \%$ percent without any change in total hours worked in a year \cite{ford2015rise}. This has been enabled through education and automation. Automation is based on sensors for gathering information which needs a communication infrastructure. Machine to machine (M2M) communications is becoming a possible solution through cheap hardware.  However, the scale of the deployed sensors is beyond the capacity of the current wireless networks \cite{hasan2013random}.

The scaling problem is dubbed in the current research as massive random access problem. The reliability aspect of the sensor communication which is reflected in the 5G research as the Ultra Reliable Low Latency Communications (URLLC). Reliability is defined as the percentage of the devices that obtain the requested service. Part of the research is focusing on solving the random access challenge reliably \cite{lien2012cooperative} within a bounded delay, resulting in a massive reliable random access problem. Both problems combined can be reformulated in resource management terms as delay guarantees for \textit{massive} number of users.  

The performance analysis for random access is built on certain assumptions about arrival traffic. This constrains the analytical guarantees such that they are only valid with the assumed traffic distribution. Exposed to a different traffic distribution, the system maybe considered instable in terms of a performance metric e.g., delay.

In order to provide traffic agnostic guarantees the protocols should be able to react to different arrival distribution. We can explore the reactivity of a protocol against a chosen metric such as stability. Tree algorithms for instance have a stable throughput with respect to various arrival traffic. However, the same cannot be said for delay such that it scales linearly with the number of users added to the system. High delay is not tolerable for some applications, and should be bounded. Protocols that solve such a bounded delay problem are required. 


In this work we present Admission Control based Traffic Agnostic Delay Constrained Random Access (AC/DC-RA) protocol, that provides stochastic delay bounds. The \textit{Traffic Agnosticism} is achieved via separation of the backlog and initial access. \mgrevv{\textit{Stochastic Delay Constraint} is enabled} through an \textit{Admission Control} decision that is based on a novel collision multiplicity estimation algorithm.
Our contributions are four-fold: 
\begin{enumerate}
\item We use a novel admission control decision, that takes place before the contention resolution (Sec.~\ref{sec:admission}). This enables guarantees for traffic agnostic \mgrevv{stochastic delay constraints} for random access.
\item A novel scalable collision multiplicity estimator is provided that is based on the famous Coupon Collector's Problem (Sec.~\ref{sec:esti}).
\item We make use of a Parallel Multi-Channel Tree Resolution that re-arranges exploration of the contention slots in order to achieve stochastic delay bounds (Sec.~\ref{sec:mpcta}).
\item We provide a dimensioning model for the suggested AC/DC-RA protocol which provides optimized use of system resources (Sec.~\ref{sec:analysis}).
\end{enumerate}

\subsection{Notation}

The sets are denoted with calligraphic capital letters $\mathcal{A}$. Sequences are denoted with bold lower-case letters $\mathbf{a}$. Sequences of sequences are denoted with bold upper-case letters $\mathbf{A}$. $\text{E}[.]$ is used for expectation and $e$ denotes the natural exponent. $\hat{(.)}$ is used for the estimated quantities.

\section{System Model}
\label{Sec:sys}
 The topology is star with a central station. The traffic is assumed to be uplink only. \mgrev{There are $N_{max}$ total users and $N_t$ active users at a time-instance $t$. The traffic of the uplink communication is sporadic, thus the set of active users is unknown to the receiver.} We base our resource model on a two dimensional grid like in an OFDMA system where one dimension is frequency and the other is time. We define, each one of the cells, as a resource. \mgrev{The resources are used on a contention basis.} \mgrevvv{Due to the broadcast nature of the wireless communication each sensor can access the same resource at the same time and interfere with each other. This phenomena is called a collision. This behaviour can be abstracted with a model.} We use the collision channel model i.e., \mgrevvv{resources} have 3 distinct states, idle (0, no request), singleton (1, 1 request) or a collision ($e$, $>$1 requests). \mgrev{{Unless physical layer enhancements are assumed the central entity cannot differentiate two or more users and treat them equally.}} We assume no capture or interference cancellation capability. We also assume an instant and costless feedback. Implementation of such feedback channels is discussed in previous work \cite{gursuslotted}.  We will use the term backlogged user for the collided users and initial arrival for the first attempt.

We define Quality of Service (QoS) as the reliability ($R$) that a packet is received at the destination within a certain delay constraint ($L$) after it is generated. We denote a set of sensors that have the same QoS requirement as class $j$ and its reliability requirement as $R_j$ and delay constraint as $L_j$. The delay $L$ incorporates delay stemming from re-transmissions due to collisions and reflect the performance of the random access channel. Any delay stemming from channel fading is not considered in this paper and only a radio resource perspective is evaluated.

We will use the term outer protocol\footnote{The inner and outer protocol was introduced to us, through a reviewer for one our previous papers. It is used initially in PDFSA paper \cite{barletta20110} of Barletta et. al. } for the traffic shaping part of the protocol that is achieved via the admission control and the term inner protocol for the contention resolution part.

\subsection{Proposal}
\begin{figure}[t!]
	\centering
	\begin{tikzpicture}[->,>=stealth',shorten >=1pt,auto,node distance=2.5cm,
	semithick, scale = 0.4, every node/.style={scale=0.6}]
	\pgfdeclarelayer{background}
	\pgfdeclarelayer{foreground}
	\pgfsetlayers{background,main,foreground}
	\tikzstyle{startstop} = [rectangle, rounded corners, minimum width=1.5cm, minimum height=1cm,text centered, draw=black, fill=TUMBlauDunkel!30]
	\tikzstyle{io} = [trapezium, trapezium left angle=70, trapezium right angle=110, minimum width=3cm, minimum height=1cm, text centered, draw=black, fill=TUMBlau!30]
	\tikzstyle{process} = [rectangle, minimum width=3cm, , text width=3cm, minimum height=1cm, text centered, draw=black, fill=TUMBlauHell!30]
	\tikzstyle{decision} = [diamond, minimum width=3cm, minimum height=1cm, text centered, draw=black, fill=TUMBlauMittel]
	\tikzstyle{arrow} = [thick,->,>=stealth]
	\node (start) [startstop] {Activate};
	\node (pro1) [process, below of=start] {Select \\ Admission Channel};
	\node (pro2) [process, below of=pro1] {Select \\ Resource Block};
	\node (coll) [decision, below of=pro2, yshift=-1cm] {Collision};
	\node (adm) [decision, below of=coll, yshift=-1cm] {Admission};
	\node (pro3) [process, below of=adm, yshift=-0.9cm] {Wait Feedback \\ for Tree Resolution};
	\node (succ) [startstop, right of=adm,xshift=2cm] {Success};
	\node (fail) [startstop, left of=adm,xshift=-2cm] {Fail};
	\begin{pgfonlayer}{background}
	\path (pro1.west |- pro1.north)+(-0.4,0.2) node (a) {};
	\path (pro2.south -| pro2.east)+(+0.4,-0.7) node (b) {};
	\path[fill=TUMGrau!20,rounded corners, draw=black!50, solid]
	(a) rectangle (b);
	\path (pro3.west |- pro3.north)+(-0.4,0.2) node (a) {};
	\path (pro3.south -| pro3.east)+(+0.4,-0.7) node (b) {};
	\path[fill=TUMGrau!20,rounded corners, draw=black!50, solid]
	(a) rectangle (b);
	
	\end{pgfonlayer}

	\draw [arrow] (start) -- (pro1);
	\draw [arrow] (pro1) -- (pro2);
	\draw [arrow] (pro2) -- (coll);
	\draw [arrow] (coll) -- node[anchor=west] {Yes} (adm);
	\draw [arrow] (coll) -| node[anchor=west] {No} (succ);
	\draw [arrow] (adm) -- node[anchor=west] {Admitted} (pro3);
	\draw [arrow, bend left] (pro3) -| (succ);
	\draw [arrow] (adm) --node[anchor=south] {Reject} (fail);
	\node[below of=pro2,yshift=0.6cm, xshift= 1cm] (layerprot) {\fontsize{8}{8}\selectfont Outer Protocol};
	\node[below of=pro3,yshift=0.6cm, xshift= 1cm] (layerprot) {\fontsize{8}{8}\selectfont Inner Protocol};
	\end{tikzpicture}
	
	\caption{AC/DC-RA Flow Diagram - Sensor perspective}
	
	\label{fig:flow_diagram}
\end{figure}
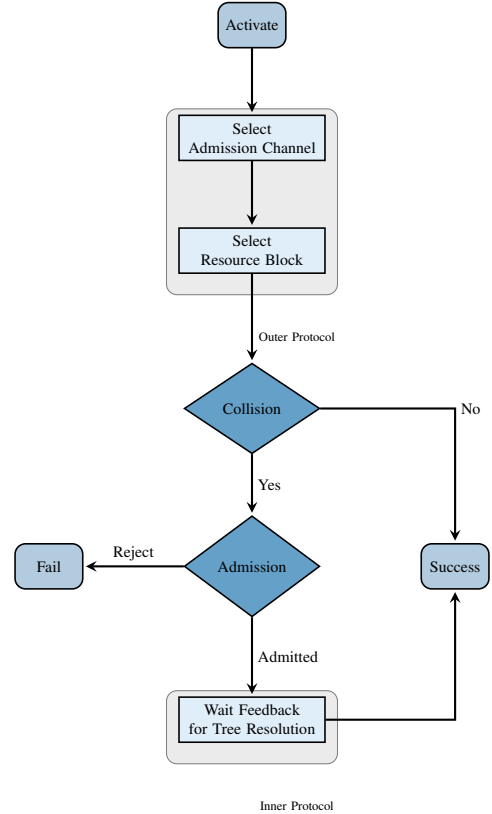
We propose an outer protocol that separates the initial arrivals from backlogged users. And an inner protocol that resolves each set of backlogged users in an isolated manner. \mgrev{The outer protocol is used for initial arrivals only. Users may collide through the use of the outer protocol. An admission decision is given for the collided users through the outer protocol. If admitted, the collided users access the inner protocol. The outer protocol uses an Admission Channel (AC) and the inner protocol uses a Resolution Channel (RC). } 

The admission is based on the \mgrevv{stochastic delay constraint} of the user, the collision multiplicity and the available capacity of the resolution channel. In the following, we explain in detail how this decision is taken. \mgrev{{The admission decision is only for the resolution resources i.e., contention resources. We do not consider a contention free resource admission scenario and it is left for future work.}} 

The set of resources $\mathcal{M}_{AC}$ and $\mathcal{M}_{RC}$ form Admission Channel (AC) and Resolution Channel (RC) respectively where ${M}_{AC} + {M}_{RC} = 	{M}$ being the total number of resources with ${M}_{x} = |\mathcal{M}_{x}|$ denoting the cardinality of the set. There may be multiple admission channels with respect to each QoS class $j$ denoted as $\mathcal{M}_{AC_j}$ and $\sum {M}_{AC_j}= {M}_{AC}$ .

\mgrev{This protocol can be summarized with a flow diagram as given in Fig.\ref{fig:flow_diagram}}. When an event notification is received, the user is activated and starts using the outer protocol. \mgrev{ It selects the \textbf{admission channel} that is appropriate for the QoS class. There are more than one admission channel so that the system can infer the QoS from the channel. Then it selects one of the \textit{resources} in that channel. This selection is done with pre-set probabilities known to the user. It transmits the packet using that resource. This terminates the outer protocol. The outcome of the transmission can be a success or a collision. The central entity observes the outcome for that resource. If it is a success, the user is informed via a broadcast and it goes back to sleep mode. If a collision occurred, then an \textbf{admission control decision} is taken for that resource by the central entity. All users that have used that resource are either rejected or admitted and informed via a broadcast feedback. In case of a rejection, a user may have another radio interface. Or the sensor can report the failure to higher layers and trigger higher layer solutions e.g. switch to local control. In case of an admission, the inner protocol is initiated. The inner protocol used is a binary tree algorithm such that after each collision users have to re-select one of two new resources. The users are informed about the resources via a broadcast feedback. The feedback and the allocation method of these resources guarantee that all admitted users are successfully resolved by the inner protocol before the delay constraint.}

\begin{figure}[!t]
	\centering
	\includegraphics[width=0.44\textwidth]{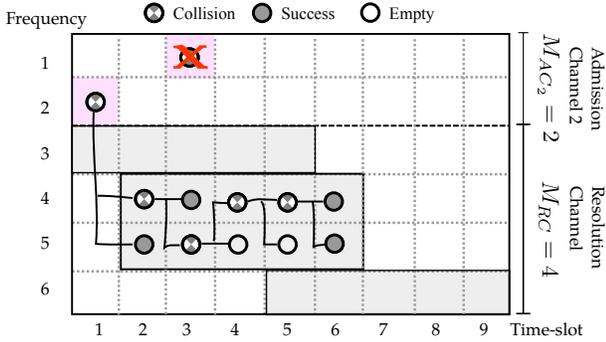}
	\caption{ Resource separation between the inner and outer protocol of AC/DC-RA and the story of a set of requests that selected the same resource through AC/DC-RA protocol.}
	
	\label{fig:srawg}
\end{figure}
 An example for the resource allocation is illustrated in Fig.~\ref{fig:srawg}. \mgrev{The illustration shows the allocation of Admission Channel and Resolution Channel resources on the resource grid, where the horizontal axis represents time and the vertical axis represents frequencies. Resource use is colored in pink for initial access and in gray for backlogged access. The boxes depict the resources allocated to the respective resolution in RC. The rejection of an initial access is depicted with a red cross. For clarification of the example ternary outcome $(0,1,e)$ of the resource use is illustrated with different symbols.} A user that requires to report a fire within $100$ ms with $0.99$ reliability selects the admission channel $2$ that represents its QoS requirement in this case. It transmits the data packet on frequency two at time-slot 1 which is a resource part of that Admission Channel. The outcome is a collision as other users have selected the same resource. Then the number of users that accessed this resource is estimated. The admission control decides that the resolution is possible within the delay constraint ($100$ ms) with the given reliability. The delay constraint is represented with 6 time-slots in the example. The admission control calculates the number of frequencies as $2$ needed for parallelization. Then it checks if it is possible to allocate $2$ frequencies in the resolution channel for that resolution. As there is available capacity in RC, the user in our example and all the other users that have collided with it are admitted to RC and resolved with a tree resolution. \mgrev{The allocated resource grid for the resolution is illustrated as a gray box limited with the lines. Within this gray box the collided users make a random selection on each time-slot bound to frequency 4 and 5. On time-slot 2, three of the users have selected the frequency 4 while one user have selected the frequency 5. In this case outcome on frequency 4 is a collision and on frequency 5 is a success. The collided users re-select one of the frequencies randomly again on the time-slot 3. This time one user has selected the frequency 4 and two users have selected the frequency 5. This results in another success. Two of the users still need to be resolved. Thus, process continues until time-slot $6$ where both of the users have selected their own resource. The resolution is completed before the delay constraint as guaranteed by the admission decision.} In the meanwhile another sensor is rejected at time-slot $3$ since required capacity is not available in RC. 

\section{AC/DC-RA - Outer Protocol}
\label{sec:out}

The outer protocol is used for the initial access of the devices.  We do not use any collision avoidance mechanism to avoid delay before any user can reach the system. 


Our proposal is based on two design choices. First, there are multiple Admission Channels and the user should select the one that is appropriate for the \textit{Quality of Service class}. Second, we customize the resource selection probabilities within any of the Admission Channels to enable collision multiplicity estimation for arbitrary number of active sensors. 
Lastly, the Outer Protocol is terminated when this information is transferred to the admission control which is the gateway between two sub-protocols.

\subsection{Separate Admission Channels - QoS Information}
\label{sec:qos}
We assume that all the devices have gone through an initial connection establishment or have overheard a broadcast. Through this information exchange, each device is aware of the appropriate admission channel for the required QoS. 

 There are multiple ACs for the initial access for different QoS classes, such that all the users in the same AC require the same delay bound and reliability. The AC is a set of resources $\mathcal{M}_{{AC}_{j}}$ e.g. for QoS class $j$. Sum of all the resources \mgrev{orthogonal to time} in the admission channels results in cardinality of the admission channels ${M}_{AC}$. \mgrev{As detailed in Section \ref{Sec:sys} a resource represents a single cell in the resource grid. }

\mgrevvv{We assume that a slot size is fixed and the bandwidth of a slot matches the payload size for each specific class. As each class is using a fixed AC, it can be expected that each AC has a unique bandwidth matching the payload size. Different classes can co-exist as different slices in the same resource grid for heterogeneous bandwidths. The possibility of this approach is investigated in 5G standardization under the bandwidth parts topic \cite{jeon2018nr}. Bandwidth parts enable co-existence of different payloads through adjusting the bandwidth of a slot. For the rest of the paper we assume homogeneous payload size among different classes and the effects of heterogeneous bandwidths are not investigated.}


\subsection{Resource Selection Probabilities - Collision Size Estimator}
\label{sec:esti}

We assume that a set of sensors of size $N_t$ at time instant $t$, selects randomly one resource from a set of resources in the admission channel at the same time. Depending on this selection a sensor may collide or be successful. Also some resources maybe unoccupied. The central entity can only observe the ternary outcome (0, 1, e) of these resources. From this outcome it has to make a collision size guess. 

A similar estimation problem has already been investigated in the state of the art for RFID tag readings \cite{kodialam2006fast} for throughput optimization. However, the estimation time scales at best linearly with the number of sensors $N_t$. However, the work relies on Poisson approximation that is valid only with high number of resources. Usually, such resources are scarce and costly in terms of delay. To solve this problem, another work has considered the resource selection probabilities as a design parameter trading off precision for estimation speed \cite{stefanovic2013joint}. Here, we aim at generalizing such an estimation to any number of active sensors and map it to the well-known Coupon Collector's Problem (CCP). 



 


\subsubsection{Coupon Collector's Problem}
There are $M$ unique coupons that are obtained through independent draws from an urn with replacement. The problem is to find the expected number of draws until all $M$ coupons are collected. Coupons may have equal or unequal selection probabilities. We will refer to selection probability  of the $i^{\text{th}}$ coupon as $p_i$ such that, 

\begin{equation}
1 = \sum_{i=1}^{M} p_i.
\label{eq:ccp_const}
\end{equation}

 This problem is solved for equal and unequal coupon selection probabilities \cite{adler2003coupon}. We do not focus on expected number of draws until all $M$ coupons are collected, but we will focus on the expected number of draws  given a certain set of uniquely drawn coupons $\mathcal{M}_{s+c}$. Thus, we are guessing the expected number of draws that have been made given that a certain set of unique collected coupons.

\subsubsection{Analogy to Collision Size Estimation}

 We define a contention in a single time-slot $t$ as an experiment. Suppose there are $N_t$ sensors selecting $M$ resources randomly on a contention basis at time-slot instance $t$. We observe the outcome of the contention on these $M$ resources. We define the outcome on a contention resource $i$ as $o_i$ where a sequence of outcome is  $\mathbf{o} =(o_1,o_2,o_3,o_4)=(1,0,e,1)$ for an example with $M=4$. The ternary outcome $o_i \in \{0,1,e\}$ of the contention for resource $i$ is converted to the set of coupons collected. We consider idle resources as not-selected coupons, i.e., the $\mathcal{M}_{s+c}$ can be defined as,
 
 \begin{equation}
 i
 	\begin{cases}
 	 \in \mathcal{M}_{s+c} \text{ if } o_i \neq 0 \\
 	 \notin  \mathcal{M}_{s+c} \text{ if } o_i = 0.
 	\end{cases}
 \end{equation}
   Using this set we calculate the expected number of draws $E[Z]$, corresponding to the estimated number of sensors at time-slot $t$ $\hat{N}_t$.
  Then, the set of selected coupons can be written as $\mathcal{M}_{s+c} = \{1,3,4\}$ since resource $2$ is idle. Using the probability of selecting any of the $M$ resources.
 




The estimated number of active sensors $\hat{N}_t$ is given with expected number of draws given a set of uniquely drawn coupons with unequal probabilities 

\begin{equation}
	\hat{N}_t  = \Exp{Z | \mathcal{M}_{s+c}} = \sum\limits_{z=0}^{\infty} \left( 1 - \prod\limits_{i \in \mathcal{M}_{s+c}}(1-e^{-p_i z}) \right),
	\label{eqn:ccp_ue}
\end{equation}
where the probability that a sensor did not select a resource $i$ is multiplied for each resource for $z$ sensors. Then this is subtracted from one to calculate the probability that all of these resources are  selected at least once. Then the expectation is taken over $z$. The sum is up to infinity to calculate the probability of an outcome given there are up to infinite sensors. For large enough $z$, probability that a resource is not selected converges to $0$. So Eq.~\eqref{eqn:ccp_ue} gives us the expected number of sensors given the outcome.  Further explanation for Eq.~\eqref{eqn:ccp_ue} is given in App.~\ref{app:proof1}. 


It is clear that each different selection $\mathcal{M}_{s+c}$ may give a different result in terms of number of sensors. We define the highest expected number of sensors as $\Exp{Z | \mathcal{M}_{s+c} = \mathcal{M}} = \hat{N}_{max}$ for, $\mathcal{M}$, the outcome of the complete set, i.e., $\forall\,\,i\,\, o_i \neq 0$, where we have a collision or success on all resources. 
The estimation range for the number of active sensors $N_t$ is up to $\hat{N}_{max}$. Therefore, the resource selection probabilities $p_i$ should be adjusted, such that $\hat{N}_{max}$ is larger than the worst case number of sensors. On the other hand it is intuitively clear that adjusting $p_i$ to increase $\hat{N}_{max}$ results in further decrease in precision of the estimation. Otherwise we can decrease $p_i$ to increase $N_{max}$ to infinity. 




\begin{table}[!t]
	\centering
	\begin{tabular}{ c || c |c| c| c| c }
		 Dist. &   Geom & Pois. & $p^0= 10^{-2}$ &  $p^0= 10^{-3}$ &  $p^0= 10^{-4}$   	 \\\hline
		  $\hat{N}_{max}$ & 	$10^2$ &	$2\cdot 10^3$&	$1.5\cdot 10^2$&	$1,1\cdot 10^3$& $9 \cdot 10^3$	
	\end{tabular}
	\caption{Expected number of draws for Coupon Collector's Problem with $M=18$ for various distributions.}
	\label{fig:mean_plot_ccp}
\end{table}

 In Table.~\ref{fig:mean_plot_ccp} we have summarized $\hat{N}_{max}$ with different distributions of $p_i$. We have used the constraint in Eq.~\eqref{eq:ccp_const} in order to \mgrev{calculate $p_i$ for various distributions. The $p_i$ for each distribution is as follows: }(1) for geometric distribution with a fixed $p$ we set the selection probability as $p_i = (1-p)^i \cdot p$, (2) for Poisson distribution with a mean $\lambda$ we set the selection probability as $p_i = \frac{\lambda^i e^{-\lambda}}{i!}$, (3) for power series, defined the selection probability as $p_i = p^0 \cdot \alpha ^i$. We have to set $p^0$ and adjust $\alpha$ accordingly. \mgrev{We then used the Eq.~\eqref{eqn:ccp_ue} to calculate $\hat{N}_{max}$.} In Table.~\ref{fig:mean_plot_ccp} we see that $p^0 \approx \frac{1}{N_{max}}$. Thus, using the power series we can easily adjust the estimation.

\subsubsection{Collision Size Estimation}

 After we have the estimated number of active devices $\hat{N}_t$, we will use the maximum likelihood to partition these devices into each resource. In the following parts we will use $N_t$ instead of $\hat{N}_t$ for ease of reading.
 
 The problem is now to partition $N_t$ devices to $M$ bins. The partitioning is constrained with the outcome $\mathbf{o}$, i.e., collision on resource $2$ and success on resource $5$ translates in to $o_2 = e, o_5= 1$. Possible guesses $\mathbf{g}$ will be sequences that fulfills the outcome constraints. The guess of resource $i$ in the $x^\text{th}$ sequence is  $g_i^x$. We also use $g_i$ for a guess for resource $i$, and $\mathbf{g^x}$ as the guess sequence $x$. Now we can write the constraints

\begin{eqnarray}
 	g_i    \label{eqn:limits}  
\begin{cases}
 	=0 & \text{if } o_i = 0 \\
 	=1 & \text{if } o_i = 1 \\
 	\geq 2 \text{ , } \leq \left( N_t -\sum_{j=1}^{i-1}g_j \right)& \text{if } o_i = e.
 	\end{cases}  	
\end{eqnarray}

We define the guess set $\mathcal{G}$ such that it involves all guess sequences fulfilling a given outcome sequence $\mathbf{o}$ and the number of active devices $N_t$. For example, with $M= 3$ and a outcome sequence of $\mathbf{o}=(o_1=1,o_2=e,o_3=e)$ where we have $N_t= 7$ we will have  $\mathcal{G} = \left( (1,2,4), (1,3,3),  (1,4,2)\right) = \{ \mathbf{g^1} , \mathbf{g^2}, \mathbf{g^3}\}$, such that $g^2_3 = 3$ and $\mathbf{g^2} = \{1, 3, 3\}$. 

We can calculate the probability of a guess as in 
\begin{equation}	
P_{\mathbf{g}} = \prod_{i \in \mathcal{M}} \left( {{N_t - \sum_{j=1}^{i-1} g_j}\choose{g_i}}\left(p_i \right)^{g_i}\right).
\end{equation}

This will enable calculation of the most likely partition, to have an estimate on how many sensors are on each resource as 
\vspace{-0.3cm}
\begin{equation}
	\hat{\mathbf{u}} = \arg\max_{\mathbf{g}} P_{\mathbf{g}} , \forall\,\, \mathbf{g}\in \mathcal{G},
\end{equation}
where $	\hat{\mathbf{u}}$ is the sequence for the collisions size estimation for all resources. \mgrev{The equation is complex to calculate with increasing dimensions of $\mathbf{g}$ as it is a combinatorial maximum likelihood calculation. It depends on $N_{max}$ and cardinality $|\mathbf{g}|$ such that ${N_{max}}^{|\mathbf{g}|}$ cases may be evaluated depending on the feedback. For practical implementations a heuristic estimator can be used an example is as such \vspace{-0.2cm}
	\begin{equation}
	\hat{u_i} = 
	\begin{cases}
	\lceil p_i \cdot \hat{N} \rceil & o_i = e\\
	o_i & o_i \neq e
	\end{cases}
	, 
	\label{eq:heur_ccp}
	\end{equation}
	where $\hat{N}$ is the total backlog estimation given by Eq.~(3) that uses the outcome sequence $\mathbf{o}$ and $N_{max}$ as input }

\subsubsection{Comparison}

As a comparison for our estimation technique, we choose two maximum-likelihood estimators (MLE). First one is based on the observation of non-idle resources only $M_{s+c}\triangleq\sum_{i=1}^{M}\mathbbm{1}_{o_i\geq1}$ (that is, without knowledge of the number of idle resources), where $\mathbbm{1}$ is the indicator function . The MLE operates on the following exact probability of observing $M_{s+c}$ non-idle resources, given a total of $M$ resources and a total of $N_t$ sensors:

\begin{eqnarray}
P_{\text{MLE}}[M_{s+c}|M, N_t] &=& \frac{\stirling{N_t}{M_{s+c}}M!}{M^{N_t}(M-M_{s+c})!},\nonumber\\
\hat{N}_t &=& \arg\max_{N_t} P_{\text{MLE}}[M_{s+c}|M, N_t]
\label{eqn:mle_dumb}
\end{eqnarray}
where $\stirling{N_t}{M_{x}}$ are the Stirling number of the second kind.

Second comparative technique is adaptation of the work from Zanella~\cite{6134704} on the RFID collision set estimation. The work is based on observing the number of collided $M_c$ and successful $M_s$ resources, and, using the approximation of the exact expression, computes the maximum-likelihood $N_t$ by finding the roots of the expression, i.e. finding the number of resources that maximizes the idle likelihood while minimizing the collision likelihood as in:

\begin{equation}
\frac{N_t-M_s}{M_c} = \frac{\frac{N_t}{M}(e^{\frac{N_t}{M}}-1)}{e^{\frac{N_t}{M}}-1-\frac{N_t}{M}}.
\label{eqn:mle_zanella}
\end{equation}

The average collision size is then computed from $\hat{N_t}$ as in $\frac{\hat{N_t}-M_s}{M_c}$. It has to be noted that, since neither of MLE approaches vary the resource selection probabilities (i.e., both use uniform probabilities), none of them can give a reliable estimate above a certain total number of active devices $N_{max}$, i.e., whenever $M_c=M$ is observed.


\begin{figure}[!t]
	\centering
	\includegraphics[width=0.33\textwidth]{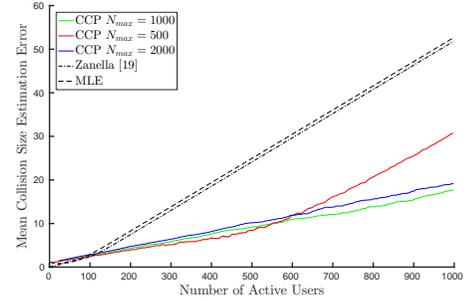}
	\caption{Mean collision size estimation error}
	\label{fig:collisionsize}
\end{figure}

\mgrev{We have conducted Monte Carlo simulations in MATLAB for comparing the estimators. The resource selection probabilities are set with respect to power distribution calculated in section \ref{sec:esti} for CCP and $N_{max}$ values are set as $500$, $1000$ and $2000$. The resource selection probabilities are set uniformly for the baseline case. The reason for this selection is that the state of the art uses the Poissonization of the outcomes which is a valid approach only with equal resource selection probabilities. }In Fig.~\ref{fig:collisionsize} collision size estimation error is plotted with 18 resources $M$. \mgrev{The absolute estimation error is calculated as $|\hat{u_i} - u_i|$ taking the difference between estimated and actual number of users per resource $i$, which is then averaged as in $\expectation[|\hat{u_i} - u_i|]$ over multiple runs and multiple resources.} CCP is compared against the state of the art with varying the number of active devices from $1$ to $1000$. Each number of active users are simulated for 1000 runs. The limitation of uniform resource selection is observed from the results. \mgrev{The MLE estimator saturates with $M=18$ after 100 users since the observation is always a set of collisions when the resource blocks have equal probability to be accessed. Thus, the MLE estimates 100 users with full collision set and the error linearly grows with the number of active users. In CCP, with increasing number of users an idle occurs and this enables scalability up to $N_{max}$ active users. }
 
 \mgrev{We have also evaluated an error in the setting of $N_{max}$ and how such a wrong setting will affect the system in Fig.~\ref{fig:collisionsize}. The $N_{max}$ set to $500$ represents the case where we may have more users accessing the medium than the allowed maximum. We see that the absolute estimation errors are almost the same up to $500$ active users. After this point the estimation error grows linearly with increasing number of users similar to the state of the art. On the other hand the case where $N_{max}$ is set to $2000$ represents that we always have a higher limit for maximum number of users compared to active number of users. This has a less critical effect compared to setting a lower maximum limit. This can be observed in Fig.~\ref{fig:collisionsize} where the absolute error has increased slightly but is in general lower compared to the previous case. Thus, it can be concluded that a relatively high $N_{max}$ can be selected to avoid the saturation effect.}

\mgrev{The scalability comes with the cost of precision loss with low number of active users. Even though, the mean error difference is approximately 1 user up to 200 active users, the state of the art is better than the CCP. This is due to setting the unequal access probabilities for scalability that is enforced due to limited amount of estimation resources. }




\mgrev{The precision of the estimation is evaluated on average. Thus, the strictness of the \mgrevv{stochastic delay constraints} provided through the use of the estimator is valid on a set of realizations, but not for each realization of the random process. Also, the \mgrevv{stochastic delay constraint} would be valid if the number of arriving users is upper-bound so the exact estimation can be converted to an upper-bound for reliability. \mgrevv{We enable this via adding the mean estimation error $\expectation[|\hat{u_i} - u_i|]$ from the analysis as a pessimism factor on top of the collision multiplicity estimate. This makes sure that the \mgrevv{stochastic delay constraint} is not violated due to estimation error.} We have evaluated the results for the guarantees where the estimator is integrated in the system in Sec.~\ref{sec:eval}.}

The outcome of the estimation and the QoS requirement is obtained from the initial access of the sensors to the admission channel. Given these information the delay of the contention tree resolution can be obtained through stochastic analysis. This information enables the admission control decision. In the following section we investigate the stochastic delay analysis of the inner protocol.

\section{AC/DC-RA - Inner Protocol}
\label{sec:inner}
In this section, we first introduce the inner protocol and quickly move on to the investigation of the stochastic analysis for delay constraints.

We deploy a version of binary tree resolution algorithm for isolated resolution of each contention. Instead of a distributed decision as usually the case for tree resolution we assume a centralized decision. In a distributed version, users select the contention resource with respect to the outcome of other contentions, i.e., with respect to the feedback. A central decision can allocate the respective resolution slot such that the user does not have to monitor the feedback continuously. The contention goes on until all users are resolved. Such a central decision requires breadth-first exploration of the tree. The number of required resources for a depth-first exploration is unbounded while for breadth-first it is deterministic and number of resources are exponential $2^m$ with tree level $m$. Another advantage of breadth-first is a possible exploration of multiple contention slots simultaneously if parallel resources exist orthogonal to time, i.e., multiple frequencies. We call this \textit{parallelization} of the resolution and $M_P$ denotes the number of parallel allocated resource for a resolution. \mgrev{An example with two possible tree algorithm parallelizations is given in Fig.~\ref{fig:parallelization}. The resolution starts with $8$ users and with the first split $3$ users select one resolution slot while the remaining $5$ select the other resolution slot.} The users are resolved with a parallelization of 2 and 4. In the case of parallelization of $M_P=2$ the resolution needs a capacity of $2$ frequencies for a duration of 4 time-slots to schedule all resolution slots. However, with a parallelization of $M_P=4$ the resolution needs a capacity of $4$ frequencies for a duration of 3 time-slots. Thus, required capacity increases since higher amount of parallel resources are blocked for faster resolution.

 Stochastic delay analysis for tree algorithms that use no parallelization can be found in \cite{molle1993conflict}. A parallelization of $Q$, the branch size, is investigated in \cite{janssen2000analysis}. A parallelization of $M_P$, an arbitrary factor, is investigated in \cite{gursu_mpcta}. Multichannel Parallel - Contention Tree Algorithm (MP-CTA) \cite{gursu_mpcta} provides analytic results for breadth-first parallelized explorations of the tree. \mgrevvv{The advantage of MP-CTA protocol compared to \cite{janssen2000analysis} is the ability to keep the throughput constant while increasing the parallelization as \cite{janssen2000analysis} sacrifices throughput for parallelization.} The delay analysis is based on parallelization of $M_P$ and it enables an efficient resolution mode selection for the required delay constraint. In our analysis the $M_P$ will map to parallel resources in the same time-slot i.e., with $M_{RC}$ resource in resolution channel we have a maximum possible parallelization of $M_P=M_{RC}$.

 \begin{figure*}[t!]
 	\centering
 	\begin{subfigure}{0.4\textwidth}
 		\begin{tikzpicture}[scale = 0.8, every node/.style={scale=0.8}]
 		\pgfdeclarelayer{background}
 		\pgfdeclarelayer{foreground}
 		\pgfsetlayers{background,main,foreground}
 		\tikzset{node style/.style={state, 
 				fill=TUMElfenbein,
 				minimum size=1.5cm,
 				circle}}
 		\tikzset{node style2/.style={state, 
 				fill=white,dashed,
 				minimum size=1.5cm,
 				circle}}
 		\node[node style] (s1) {  8 };
 		\node[node style, right of=s1,yshift=2.1cm,xshift=1.1cm] (s2) {5};
 		\node[node style, right of=s1,yshift=-2.1cm,xshift=1.1cm] (s3) {3};,
 		\node[node style, right of=s2,yshift=1.3cm,xshift=1.1cm] (s4) {2};
 		\node[node style, right of=s2,yshift=-0.8cm,xshift=1.1cm] (s5) {3};
 		\node[node style, right of=s3,yshift=1.1cm,xshift=1.1cm] (s6) {2};
 		\node[node style, right of=s3,yshift=-1.1cm,xshift=1.1cm] (s7) {1};
 		\draw (s1) edge[line width=0.02cm, auto=left,->] node {} (s2);
 		\draw (s1) edge[line width=0.02cm, auto=left,->] node {} (s3);
 		\draw (s2) edge[line width=0.02cm, auto=left,->] node {} (s4);
 		\draw (s2) edge[line width=0.02cm, auto=left,->] node {} (s5);
 		\draw (s3) edge[line width=0.02cm, auto=left,->] node {} (s6);
 		\draw (s3) edge[line width=0.02cm, auto=left,->] node {} (s7);
 		\node[right of=s4,xshift=2cm] (a3){};
 		\begin{pgfonlayer}{background}
 		\path (s1.west |- s1.north)+(-0.1,1.7) node (a) {};
 		\path (s1.south -| s1.east)+(+0.1,-2.0) node (b) {};
 		\path[fill=TUMGrau!20,rounded corners, draw=black!50, solid]
 		(a) rectangle (b);
 		\path (s1.west |- s1.north)+(1.9,3.3) node (c) {};
 		\path (s1.south -| s1.east)+(2.2,-3.3) node (d) {};
 		\path[fill=TUMGrau!20,rounded corners, draw=black!50, solid]
 		(c) rectangle (d);
 		\path (s4.west |- s4.north)+(-0.1,0.1) node (e) {};
 		\path (s5.south -| s5.east)+(0.1,-0.4) node (f) {};
 		\path[fill=TUMGrau!20,rounded corners, draw=black!50, solid]
 		(e) rectangle (f);
 		\path (s6.west |- s6.north)+(-0.1,0.1) node (g) {};
 		\path (s7.south -| s7.east)+(0.1,-0.4) node (h) {};
 		\path[fill=TUMGrau!20,rounded corners, draw=black!50, solid]
 		(g) rectangle (h);
 		\end{pgfonlayer}
 		\node[below of=f,yshift=1.2cm, xshift=-0.85cm] (layerprot) {\fontsize{8}{8}\selectfont $t=3$};
 		\node[below of=d,yshift=1.2cm, xshift=-0.85cm] (layerprot) {\fontsize{8}{8}\selectfont $t=2$};
 		\node[below of=b,yshift=1.2cm, xshift=-0.85cm] (layerprot) {\fontsize{8}{8}\selectfont $t=1$};
 		\node[below of=h,yshift=1.2cm, xshift=-0.85cm] (layerprot) {\fontsize{8}{8}\selectfont $t=4$};				
 		\end{tikzpicture}	
 		\caption{Parallelization with 2 frequencies $M_P=2$}
 	\end{subfigure}
 	\begin{subfigure}{0.4\textwidth}
 		\begin{tikzpicture}[scale = 0.8, every node/.style={scale=0.8}]
 		\pgfdeclarelayer{background}
 		\pgfdeclarelayer{foreground}
 		\pgfsetlayers{background,main,foreground}
 		\tikzset{node style/.style={state, 
 				fill=TUMElfenbein,
 				minimum size=1.5cm,
 				circle}}
 		\tikzset{node style2/.style={state, 
 				fill=white,dashed,
 				minimum size=1.5cm,
 				circle}}
 		\node[node style] (s1) {  8 };
 		\node[node style, right of=s1,yshift=2.1cm,xshift=1.1cm] (s2) {5};
 		\node[node style, right of=s1,yshift=-2.1cm,xshift=1.1cm] (s3) {3};,
 		\node[node style, right of=s2,yshift=1.1cm,xshift=1.1cm] (s4) {2};
 		\node[node style, right of=s2,yshift=-1.1cm,xshift=1.1cm] (s5) {3};
 		\node[node style, right of=s3,yshift=1.1cm,xshift=1.1cm] (s6) {2};
 		\node[node style, right of=s3,yshift=-1.1cm,xshift=1.1cm] (s7) {1};
 		\draw (s1) edge[line width=0.02cm, auto=left,->] node {} (s2);
 		\draw (s1) edge[line width=0.02cm, auto=left,->] node {} (s3);
 		\draw (s2) edge[line width=0.02cm, auto=left,->] node {} (s4);
 		\draw (s2) edge[line width=0.02cm, auto=left,->] node {} (s5);
 		\draw (s3) edge[line width=0.02cm, auto=left,->] node {} (s6);
 		\draw (s3) edge[line width=0.02cm, auto=left,->] node {} (s7);
 		\node[right of=s4,xshift=2cm] (a3){};
 		\begin{pgfonlayer}{background}
 		\path (s1.west |- s1.north)+(-0.1,3.5) node (a) {};
 		\path (s1.south -| s1.east)+(+0.1,-3.7) node (b) {};
 		\path[fill=TUMGrau!20,rounded corners, draw=black!50, solid]
 		(a) rectangle (b);
 		\path (s1.west |- s1.north)+(1.9,3.5) node (c) {};
 		\path (s1.south -| s1.east)+(2.2,-3.7) node (d) {};
 		\path[fill=TUMGrau!20,rounded corners, draw=black!50, solid]
 		(c) rectangle (d);
 		\path (s1.west |- s1.north)+(4.0,3.5) node (e) {};
 		\path (s1.south -| s1.east)+(4.3,-3.7) node (f) {};
 		\path[fill=TUMGrau!20,rounded corners, draw=black!50, solid]
 		(e) rectangle (f);
 		\end{pgfonlayer}
 		\node[below of=f,yshift=1.2cm, xshift=-0.85cm] (layerprot) {\fontsize{8}{8}\selectfont $t=3$};
 		\node[below of=d,yshift=1.2cm, xshift=-0.85cm] (layerprot) {\fontsize{8}{8}\selectfont $t=2$};
 		\node[below of=b,yshift=1.2cm, xshift=-0.85cm] (layerprot) {\fontsize{8}{8}\selectfont $t=1$};
 		\end{tikzpicture}
 		\caption{Parallelization with 4 frequencies $M_P=4$}	
 	\end{subfigure}	
 	\caption{Example of parallel exploration of trees.}
 	\label{fig:parallelization}
 \end{figure*}
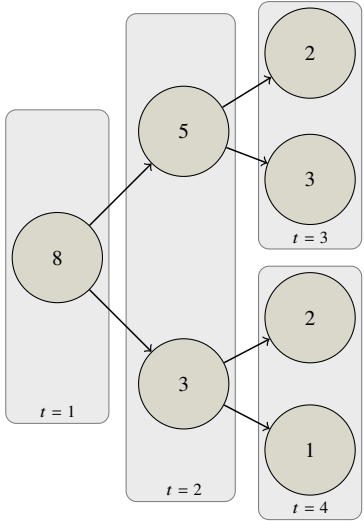
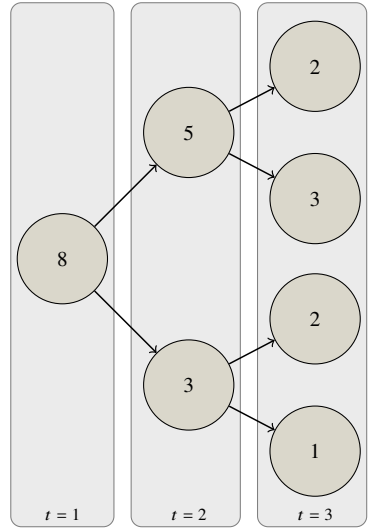




\subsection{Delay Constrained Resolution}
\label{sec:mpcta}
In this section, we investigate how the analysis for the inner protocol can be used for the admission control. 

For a \mgrevv{stochastic delay constraint} $L$ and reliability $R$, e.g., $ R = 0.95$, means that the delay constraint $L$ should be achieved $95$ percent of the time.

Stochastic delay bounds for MP-CTA are given in \cite{gursu_mpcta} for different number of sensors. These values can be placed in a look up table (LUT) for varying $N$ number of sensors, $L$ the delay, for a specific reliability $R$ as in Tab.~\ref{tab:deadline}. The LUT then outputs the minimum number of parallelization $M_P$ required to \mgrevv{fulfill} the \mgrevv{stochastic delay constraint} of all the devices in the contention resolution. If it is infeasible then it returns zero. \mgrev{For example, given 10 devices and a delay of 5 slots, it is infeasible to achieve a resolution where all devices are resolved with $0.95$ reliability. This is denoted as $M_P =0$. However, a delay of 10 slots is achievable with a parallelization of $M_P = 3$.} 

We use this analysis and define a function $f$ that outputs the number of resources $M_P$ given the required reliability and delay constraint with the number of backlogged sensors, 
\vspace{-0.2cm}
\begin{eqnarray}
	f(L_i,R_i,N)  =   \label{eqn:gra}  
	\begin{cases}
		0 & \text{if } \text{infeasible} \\
		M_P & \text{if } \text{feasible}.
	\end{cases}  	
\end{eqnarray}

Infeasibility is invoked when allocation of all the $M_{RC}$ frequencies in the resolution channel cannot achieve the required delay then $f=0$ is returned. 


\begin{table}
\centering
\begin{tabular}{c | c c c c c c c}
  Delay Bound  & 5& 10& 15 & 20 &  25 &  30 & 35 	 \\\hline
  No. Sensors &&&&&&& \\
5  & 0 & 1 & 1 & 1 & 1 & 1 & 1 \\
10 & 0 & 3 & 2 & 1 & 1 & 1 & 1 \\
15 & 0 & 4 & 2 & 2 & 1 & 1 & 1 \\
20 & 0 & 6 & 3 & 2 & 1 & 1 & 1 \\
25 & 0 & 8 & 4 & 3 & 2 & 2 & 2 \\
30 & 0 & 9 & 4 & 3 & 2 & 2 & 2 \\
35 & 0 & 11 & 5 & 4 & 3 & 2 & 2\\
40 & 0 & 13 & 6 & 4 & 3 & 3 & 2
\end{tabular}

\caption{The parallelization $M_P$, given in table, needed to resolve certain number of backlogged sensors for varying delay bounds $L_j$ and a reliability level $R_j=0.95$. The reliability level is not a dimension of the table.}

\label{tab:deadline}
\end{table}


\mgrev{We can check a concrete example using the values shared in Tab.~\ref{tab:deadline}\footnote{The values shared in the table are calculated using the analysis in \cite{gursu_mpcta}.}}. An example would be for a delay constraint of $15$ slots with $20$ backlogged sensors and a reliability of $0.95$ percent. \mgrev{We can read the cross-section of these values to see the required parallelization. }This can be formulated as \textbf{$f(15,0.95,20) =3$} such that we can use parallelization of $3$ to achieve the \mgrevv{stochastic delay constraint} in an efficient manner. The required parallelization is $2$ for $10$ backlogged sensors, and $4$ for $25$ backlogged sensors. Thus, we can allocate just the right number of resources to achieve the \mgrevv{stochastic delay constraint}. 

In this section we have shown that a delay constrained resolution is achievable through the MP-CTA. In the following section we explain how the information provided via the outer protocol will enable guarantees though use of the inner protocol, this leads us to the admission decision.




\section{AC/DC-RA - Admission Control}
\label{sec:admission}
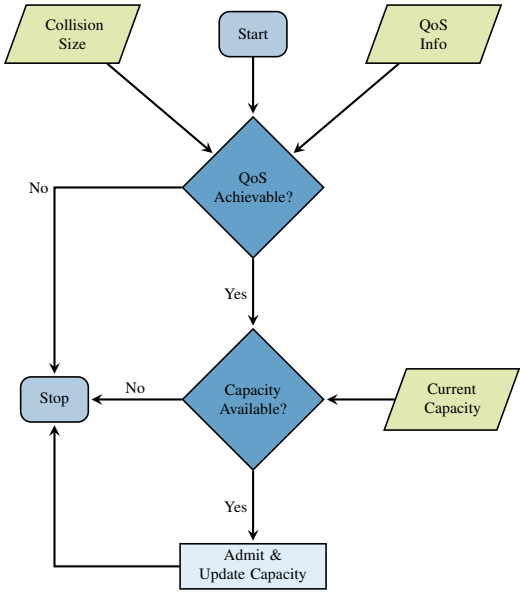
\begin{figure}[t!]
	\centering
	\begin{tikzpicture}[->,>=stealth',shorten >=1pt,auto,node distance=2.4cm,
	semithick,scale = 0.6, every node/.style={scale=0.6}]
	\pgfdeclarelayer{background}
	\pgfdeclarelayer{foreground}
	\pgfsetlayers{background,main,foreground}
	\tikzstyle{startstop} = [rectangle, rounded corners, minimum width=1.5cm, minimum height=1cm,text centered, draw=black, fill=TUMBlauDunkel!30]
	\tikzstyle{io} = [trapezium, trapezium left angle=70, trapezium right angle=110, minimum width=3cm, text width=1.2cm, minimum height=1cm, text centered, draw=black, fill=TUMGruen!30] %
	\tikzstyle{process} = [rectangle, minimum width=1cm, , text width=3cm, minimum height=1cm, text centered, draw=black, fill=TUMBlauHell!30]
	\tikzstyle{decision} = [diamond, minimum width=3cm, minimum height=1cm, text centered, draw=black, fill=TUMBlauMittel, text width=2cm,]
	\tikzstyle{arrow} = [thick,->,>=stealth]
	%
	%
	%
	\node (start) [startstop] {Start};
	\node (io1) [io, right of=start,xshift=+1.6cm] {QoS Info};
	\node (io2) [io, left of=start, ,xshift=-1.6cm] {Collision Size};
	\node (adm1) [decision, below of=start,yshift=-1.0cm ] {QoS \\ Achievable?};
	\node (adm2) [decision, below of=adm1, yshift=-2.3cm ] {Capacity Available?};
	\node (io3) [io, right of=adm2,xshift=2cm] {Current \\ Capacity};
	\node (proc1) [process, below of=adm2,yshift=-1.3cm] {Admit \& \\ Update Capacity};
	\node (stop) [startstop, left of=adm2,xshift=-2cm ] {Stop};
	
	\draw [arrow] (start) -- (adm1);
	\draw [arrow] (io1) -- (adm1);
	\draw [arrow] (io2) -- (adm1);
	\draw [arrow] (io3) -- (adm2);
	\draw [arrow] (adm1) --node[anchor=east] {Yes}  (adm2);
	\draw [arrow] (adm1) -|node[anchor=east] {No}  (stop);
	\draw [arrow] (adm2) --node[anchor=south] {No}  (stop);
	\draw [arrow] (proc1) -| (stop);
	\draw [arrow] (adm2) --node[anchor=east] {Yes}  (proc1);
	\end{tikzpicture}
	
	\caption{Admission control decision state diagram}
	
	\label{fig:flow_diagram_admission}
\end{figure}

AC/DC-RA is not improving the throughput of random access but limiting delay for a resolution. Thus, dealing with increasing number of users is still an issue. In order to investigate the scaling problem, we have to consider the capacity of the Resolution Channel.

We define \textit{capacity} as a set of resolution resources. \mgrevvv{Each collision needs different set of resources. Thus, we have to distribute the capacity in an efficient manner. Moreover, with increasing number of users in a collision we cannot scale resources in time but only in frequency as we are dealing with a delay constraint. Thus, having frequencies available is a deciding factor for whether we can resolve a collision before the delay-constraint.} \par The resources for resolution is fixed in terms of frequency and time. For instance, the \textbf{capacity} required to resolve a collision given in Fig.~\ref{fig:srawg} is a $2$ frequency $5$ time-slot grid. The $2$ frequencies are blocked for $5$ time-slots. The capacity of the Resolution Channel is also defined in terms of frequencies $M_{RC}$. It is clear that not all collisions will fit in the RC. Thus, to guarantee that users admitted to the system are always served within the \mgrevv{stochastic delay constraint}, we have to reject some of the users. The decision whether to reject the users or to admit them to RC is done by the admission control of AC/DC-RA. 

\subsection{Admission Control}
The admission is decided through evaluation of QoS information, collision size information against the resolution channel capacity. We zoom in the admission block from Fig.~\ref{fig:flow_diagram}. We provide another flow diagram for the admission control decision in Fig.~\ref{fig:flow_diagram_admission}. The \textbf{QoS information} is extracted in terms of $L_j$ and $R_j$ from the selected admission channel index $j$. The {collision size} estimation returns the vector $\hat{\mathbf{u}}$ where $\hat{u_i}$ is the \textbf{collision size estimation} for the $i^\text{th}$ resource. \mgrevv{We add the mean estimation error for the expected $N_{max}$, calculated with $\expectation[|\hat{u_i} - u_i|]$ as a pessimism factor to each collision multiplicity estimation that gives $\hat{u_i}^\dagger$. }\mgrev{As the realization $N$ is unknown, the estimation $\hat{u_i}^\dagger$  has to be used in this case for the \mgrevv{delay constrained} resource allocation calculation.} The admission control feeds this information to the stochastic tree analysis $f(L_j, R_j, \hat{u_i}) = M_{P_{j,i}}$ to obtain the number of required resources. In case the QoS is not achievable, i.e., $M_{P_{j,i}}=0$, the devices are directly rejected. If not, the requested number of resources are compared against the available number of resources in RC. If there is enough capacity the users are let into the system for resolution or else are rejected. The admission decision $D_{j,i}$ that is given for all users in resource $i$ of the admission channel $j$ can be summarized as in,
\vspace{-0.2cm}
\begin{eqnarray}
D_{j,i}  =   \label{eqn:decision}  
\begin{cases}
\text{Reject} & \text{if } M_{P_{j,i}} =0 \text{ or }  M_{P_{j,i}} > M_{RC}^t  \\
\text{Accept} & \text{if } M_{P_{j,i}} <= M_{RC}^t,
\end{cases}  	
\end{eqnarray}
where $M_ {RC}^t $ is the number of available resources in the resolution channel at time-slot $t$ and is updated as $M_{RC}^t \leftarrow M_{RC}^t - M_{P_{j,i}}$ after an accept decision. It is initialized as $M_{RC}^t =M_{RC}$ and after each resolved contention, the freed resources are added back.

\mgrevvv{Each sensor is aware of the indices of its selected resource denoted with $i$ for resource and $j$ for the admission channel. The admission decision and the resolution resources are broadcast with attaching these two indices to the decision message, such that each sensor can deduce which resources it can use for the resolution.}





The system will operate in a resource limited environment such that allocation of resources to admission channel and resolution channel will impact the behavior of the system. In order to analyze this trade-off we propose an analytical model. 

\section{Analysis}

\label{sec:analysis}
\mgrevvv{We foresee that the number of channels can be adjusted with respect to the incoming traffic. In order to analyze the effect of selecting certain number of admission channels ${M}_{AC}$ versus resolution channels ${M}_{RC}$ we} propose a Markov Chain model as given in Fig.~\ref{eq:state_machine}. We simplify the system to five different states. Initial state is an \textit{Off} state that represents the device activation characteristics with respect to the application. When active with the probability $p_{on}$, the sensor goes to the transmission state \textit{Tx}. This state is the initial access state, and the sensor selects one of the resources, $i$, in the $j^{\text{th}}$ admission channel $\mathcal{M}_{AC_j}$ and \mgrev{transmit a packet with that resource}. This selection is done on the appropriate admission channel for QoS class.

\par The initial access is a success with probability $1-p_c$. Then the sensor may go to success state \textit{Suc}. After the transmission is completed it goes back to \textit{Off} state. If the initial access results in a collision it goes to the admission state \textit{$A_R$} with probability $p_c$. In this state the number of collided sensors with that specific sensor is estimated and a decision whether resolution time is within QoS class of the sensor is given. 

\par After initial access, the sensor is admitted with probability $1-p_r$. After successful contention resolution it proceeds to the \textit{Suc} state. If the sensor cannot be admitted then it is rejected with probability $p_r$ and goes to the fail state \textit{Fail} where it informs higher layers before going to the \textit{Off} state. 

\begin{figure}[t!]
	\centering
\begin{tikzpicture}[->,>=stealth',shorten >=1pt,auto,node distance=3.2cm,
semithick,scale = 0.8, every node/.style={scale=0.8}]
\tikzstyle{every state}=[fill=TUMBlau,draw=none,text=white]

\node[state] (A)                    {$Fail$};
\node[state]         (E) [right of=A]       {$Tx$};
\node[initial, state]         (B) [above of=E] {$Off$};
\node[state]         (D) [below of=E] {$A_{R}$};
\node[state]         (C) [right of=E] {$Suc$};

\path (A) edge   [bend left]            node {$1$} (B)
(B) edge [loop above] node {$1-p_{on}$} (B)
edge              node {$p_{on}$} (E)
(C) edge   [bend right]            node {$1$} (B)
(D) edge [bend left]  node {$p_r$} (A)
edge              node {$1-p_r$} (C)
(E) edge   node {$1-p_c$} (C)
edge             node {$p_{c}$} (D);
\end{tikzpicture}

\caption{Markov Chain for AC/DC-RA}

\label{eq:state_machine}
\end{figure}
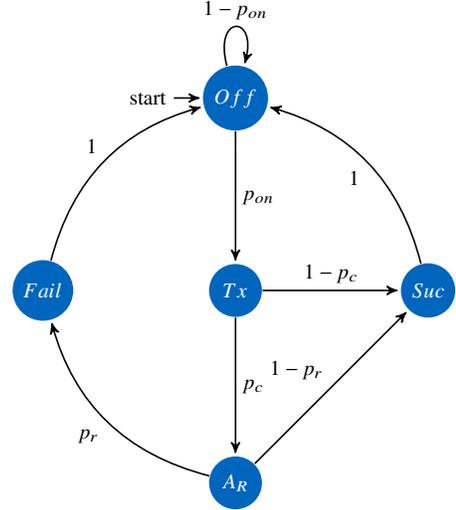





We can extract the state probabilities in terms of state transition probabilities as,

\begin{align}
P_{Off}  & = \frac{1}{1 + 3p_{on}+ p_{on}p_c\left(1 - p_r \right)} \\
P_{Tx} &  = \frac{p_{on}}{1 + 3p_{on}+ p_{on}p_c\left(1 - p_r \right)} \\
P_{A_{R}} & = \frac{p_c p_{on}}{1 + 3p_{on}+ p_{on}p_c\left(1 - p_r \right)} \\
P_{Suc} & = \frac{\left(p_r p_c - p_c p_r\right) p_{on}}{1 + 3p_{on}+ p_{on}p_c\left(1 - p_r \right)} \\
P_{Fail} & = \frac{\left(1-p_r p_c + p_c p_r\right) p_{on}}{1 + 3p_{on}+ p_{on}p_c\left(1 - p_r \right)} .
\label{eq:markov_tran_prob}
\end{align}

%
%

  


We investigate the state transition probabilities as follows: The activation probability depends on the application. For the sake of steady state analysis we consider Poisson arrivals in this scenario, which is usually assumed for sensors \cite{demirkol2009impact} \mgrevv{ and \cite{3GPP_Poisson}.}
To provide an average dimensioning we assume the probability that a device generates any packet between two random access opportunities, and the total mean arrival rate as $\lambda$ with activation probability $p_{on} = 1 - e^{-\lambda}$.

\subsection{Collision Probability $p_c$ }
For the calculation of the collision probability we cannot use typical Multichannel Slotted ALOHA models since we have modified resource selection probabilities. We solve this through modeling the problem as a bins and balls problem. The bins represent the sensors and the balls represent the resources in admission channel.
\mgrevvv{
\begin{thm}
Probability to have exactly $u$ balls out of $N$ balls in any of the $M$ bins, where each bin $i$ have an unequal probabilities $p_i$ to be land on by a ball, can be given by,
\begin{equation}
	p_c (u) = \sum_{j=1}^{J} \sum_{x \in \mathbf{S}^J_x} \left(  P_{J}^x (u) \cdot j \right),
	\label{eq:coll2}
\end{equation}
where $J$ is the maximum number of $u$-ball-groups that can be formed out of $N$ balls given there are $M$ bins. The probabilities $P_{J}^x (u)$ to have a ball-to-bin partition with $u$ balls are summed separately from one, up to and including $J$ bins with $u$ balls and weighed accordingly.
\label{thm:1}
\end{thm}
\begin{proof}
	The detailed proof is given in App.~\ref{sec:app3}
\end{proof} 
}Thus, we can calculate the probability of a collision $p_c$ as, $
	p_c = 1 - p_c(1) -p_c(0)$. After the Transmission State then we move to the Admission State.



\subsection{Admission Rejection Probability $p_r$}

\mgrev{Collisions are resolved with the tree algorithm. Each of these resolutions occupy $M_P \cdot L$ resources where $M_P$ is selected with respect to the number of collided users and $L$ is the delay constraint in terms of time-slots. As we have finite resources in our system, allocating resources to the resolutions can be considered as a serving process. Thus, }we model the serving of a resolution as a queue, Random Access Queue (RAQ), process, where each resolution resource is a server and arrivals are collisions to be served. It is a queue with no buffer since the admission decision is given instantly. \mgrev{In this section, we investigate the RAQ model in order to analytically provide the blocking probability in such a queue, that will be representing an admission rejection probability $p_r$ decision due to insufficient amount of resources in RC.}

We model each collision as an arrival to the RAQ. Since we expect a collision on all resources to use admission channel effectively, the average number of collisions can be written as $\lambda_{RAQ} = M_{{AC}_j} $ for class $j$. Thus, on heavy load, we expect a collision on all AC resources, i.e., deterministic arrivals. With low load, we expect probabilistic number of collisions thus, a Markovian number of arrival to the RAQ. 




In order to guarantee the resolution time we reserve frequencies during the contention resolution. This is necessary for modeling each resolution with a serving time. For serving time we have a deterministic value $
h_{RAQ_j} = L_j,
$
such that the serving time for each QoS class depends only on the delay constraint. \mgrev{The number of available resources is converted to the number of servers. The parallelization of the resolution is governed by the collision size.} We can calculate the expected level of parallelization as in, $\Exp{M_P}  = \sum_{u = 0}^{N_{max}} f(L_j,R_j,u) P_c(u),$ where $P_c(u)$ is the probability that a collision with size $u$ occurs and given with Eq.~\eqref{eq:coll2}. Thus, each resolution needs on average $M_P$ resources. And we have $M_{RC}$ resources in total. Via dividing the total number of resources to the average number of resources per resolution we can calculate the expected number of on-going resolutions as $M_G =\left\lfloor \frac{M_{RC}}{\Exp{M_P}} \right\rfloor$. $M_G$ represents the average number of servers in the resolution channel. 
This leads us to the admission rejection probability that can be written as the RAQ blocking probability.
\mgrevvv{
\begin{thm}
The RAQ blocking probability, given there are $M_G$ servers, the deterministic serving time of $L_j$ and $M_{AC_j}$ arrivals per slot is 
\begin{equation}
	 p_{r} = \frac{\frac{\left(L_j \cdot M_{AC_j}\right)^{M_G}}{{M_G}!}}{ \sum_{o=1}^{{M_G}} \frac{\left(L_j \cdot M_{AC_j}\right)^o}{o!}}.
	  \label{eq:pr}
\end{equation}
Here, we have projected that, if Markovian number of arrivals and deterministic number of servers swap behavior such that there are Markovian number of servers and deterministic arrivals, the same blocking probability can be used.
\end{thm}
\begin{proof}
	We leave the proof to the reader using the call blocking probabilities in \cite{gimpelson1965analysis}.
\end{proof}}
 Finally, we have all the parameters required to analyze the protocol. 
 
%
%

\section{Evaluation}
\label{sec:eval}
In this section we first evaluate the suggested algorithm in a prioritization scenario. Following this we compare our analysis with simulation results to show that the analysis provides a reasonable estimate to enable analytic dimensioning of the system. All the simulations are done in a MATLAB based discrete time simulator. 

\mgrev{We want to share certain relevant parameters considered in the simulator. We assume zero propagation time. We have implemented a collision channel model based simulator on MAC layer and perfect channel conditions are assumed. We assume costless and immediate feedback which is necessary for both tree and access barring based solutions. } \mgrevv{We investigate a single cell scenario for uplink traffic. We assume resources are organized in time and frequency.}

\subsection{Comparison with Analysis}

\begin{figure}[!t]
	\centering
	\includegraphics[width=0.33\textwidth]{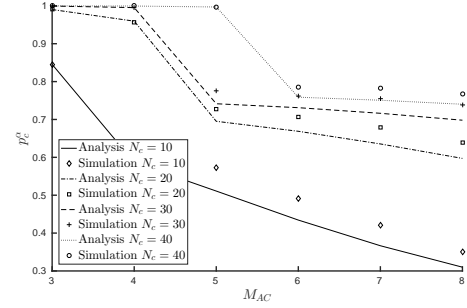}
	\caption{Evaluation of $p_c^\alpha$ with varying the resources in admission channel $M_{{AC}}$, Poisson arrivals with distinct means of $N_c = 10, 20, 30, 40$ users are evaluated.}
	\label{fig:ana_sim_comp_pc}
\end{figure}

\mgrev{We investigate the behavior of the protocol with various resource separation decisions and to show} 
validity of the analysis we simulate the AC/DC-RA with varying number of resources for admission channel $M_{AC}$ and resolution channel $M_{RC}$ and compare with our analysis. While varying the size of one channel we fix the other to $\{15,25,45\}$. We assume a Poisson arrival rate with average of $30$ active users per time-slot.

\par We compare the analysis of number of arriving collisions with simulations in Fig.~\ref{fig:ana_sim_comp_pc} where we plotted the varying number of resources in admission channel $M_{AC}$ against the collision probability that is normalized with respect to the mean Poisson arrivals. Since in simulations we use Poisson arrivals, we use the law of total probability over the probability of observing different number of devices as $
p_c^{\alpha }(N_c)  = \sum_{i = 0 }^{\infty} e^{-N_c} \frac{\left(N_c \right)^i}{i !}  p_c (i) 
$ where $p_c^{\alpha }(N_c)$ is the probability adjusted for Poisson arrivals with mean $N_c$. 

\begin{figure}[!t]
	\begin{subfigure}[b]{0.44\textwidth}
		\centering
		\includegraphics[width=1\textwidth]{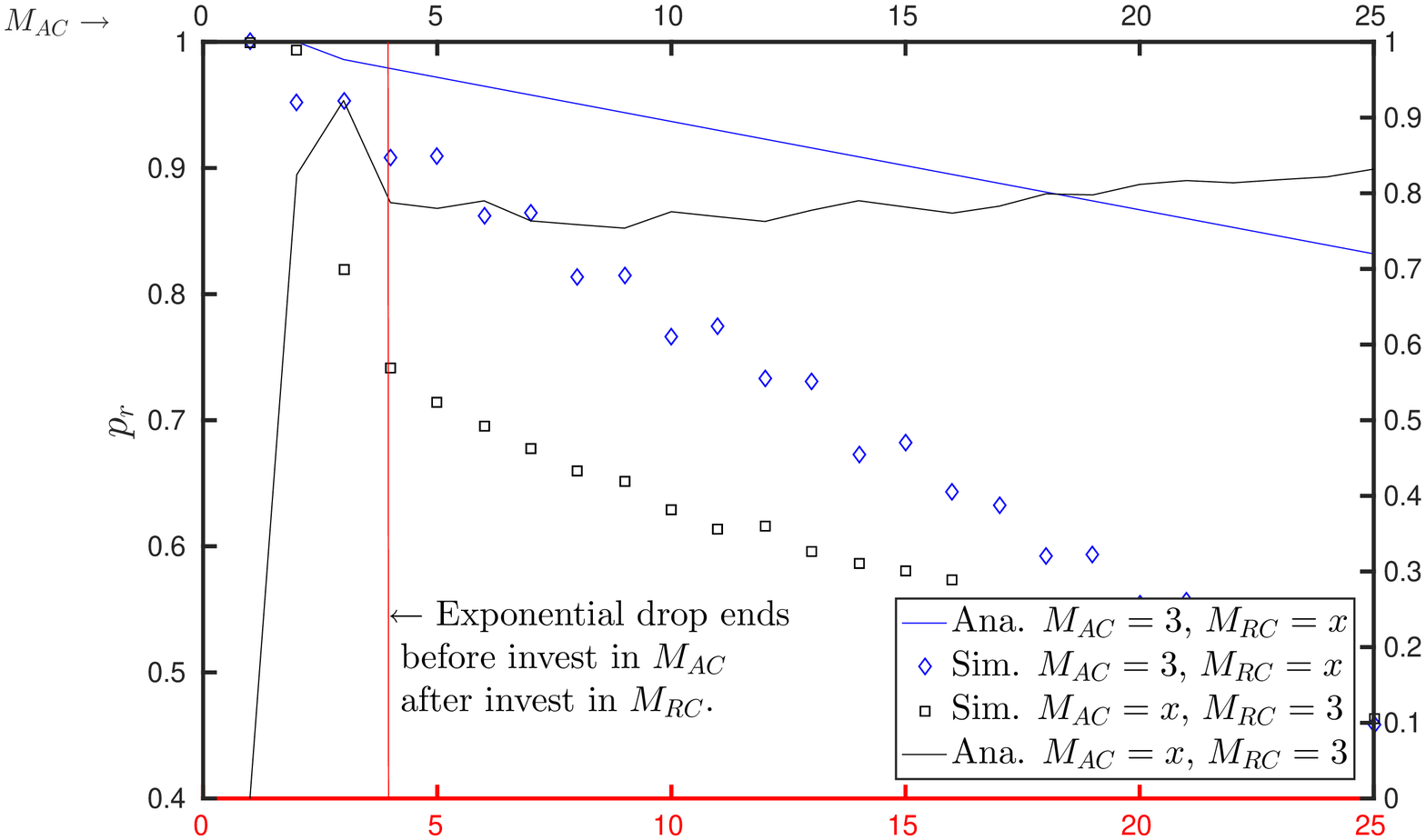}
		\caption{Fix one to $25$}
		\label{fig:ana_sim_pbraq_mj_n25}
	\end{subfigure}
	\begin{subfigure}[b]{0.44\textwidth}
		\centering
		\includegraphics[width=1\textwidth]{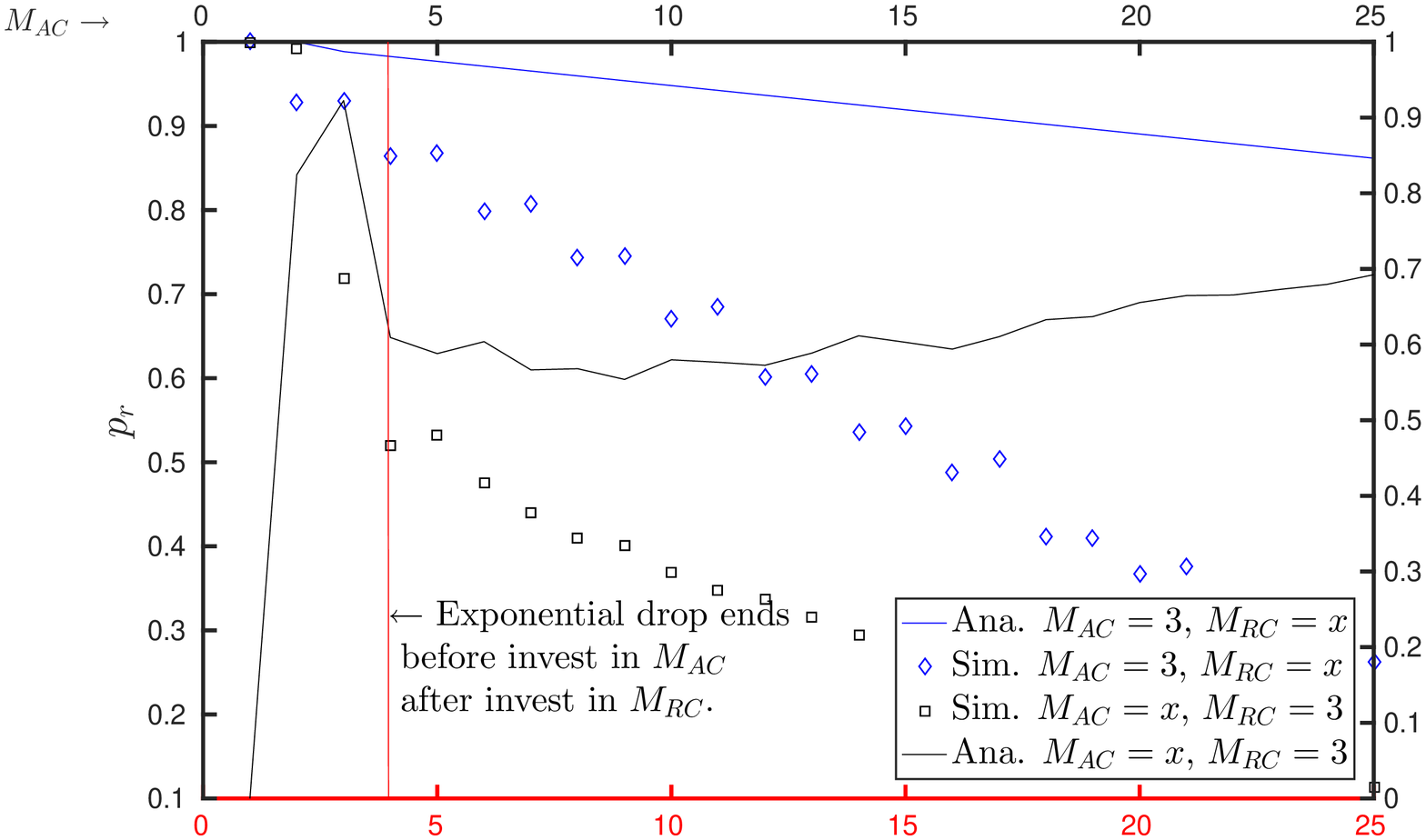}
		\caption{Fix one to $45$}
		\label{fig:ana_sim_pbraq_mj_n45}
	\end{subfigure}
	\caption{Comparison of AC/DC-RA analysis with simulations varying the amount of allocated resources to $M_{AC}$ and $M_{RC}$ with average 30 users per slot.}
	\label{fig:ana_sim_comp}
\end{figure}

We see that the with the pre-selected resource selection probabilities we can trace the collision probabilities with the given analytics. It is important to emphasize that since the complexity of the calculation grows exponentially it can only be used for offline dimensioning of the algorithm. Another observation is that the power series has a higher success rate than expected when it comes to sacrificing the throughput on the admission channel. For instance with 4 resources and 10 devices only, $60 \%$ of the resources have seen a collision on average. Since we expect 1 out of 4 resources to be free so that the estimation works on the edge, we expect around  $75 \%$ collisions. This shows that most of the users focus on only some resource such that it is possible only one user selects the last resource. \mgrevvv{The slight mismatch between the analysis and the simulations is limited to 4 percent. As the probability of a collision with $u$ users is taken into account to calculate the expected parallelization, this mismatch does only slightly affect the admission rejection probability. }


In Fig.~\ref{fig:ana_sim_comp} we varied number of resources in resolution channel and admission channel on the x-axis and we plotted the admission rejection probability on the y-axis. \mgrev{The analysis is given with a solid line while the simulation is marked with data points.} The analysis matches perfectly for low number of resources. We remind the reader that we still use the modified resource selection probabilities such that the devices will forcefully collide. Through this, we have the same number of collisions as the number of resources in the admission channel for high arrival rates. Thus, the assumption of deterministic number of \mgrevvv{arrivals} is valid for low number of resources for admission channel. However, this assumption does not hold if we a have high number of resources in the admission channel $M_{AC}$ such that more than 1 slot may be empty. \mgrevvv{The analysis is based on the assumption that high number of collisions will be observed even with increasing $M_{AC}$ therefore it is extra pessimistic. In reality, with increasing $M_{AC}$ the number of arrivals to RAQ becomes lower than $M_{AC}$ and the system can serve the collisions in a relaxed manner as shown by the difference between the analytical and simulative curves in the figure. However, the crucial point is the admission channel serves as a traffic shaper, affecting how many users will be in a collision. If the number of AC resources is too low, for instance 2, then $N$ users can be separated in two collisions with the size of $N/2$ at the best case. With a too low number of AC channel resources, most of the collisions are too big to serve them before the delay constraint\footnote{ The binary tree algorithm, even with infinite channels, has a maximum number of users it can solve that is limited with contention slots at a certain level of the tree, i.e., 2 users at level 1 and 4 at level 2. See \cite{janssen2000analysis} for maximum number of collided users that can be resolved within a certain deadline given infinite channels.}. Thus they are rejected, despite all resolution channel resources. With a certain higher number of AC resources, the infeasible collision do not or only rarely occur. As the infeasibility is avoided, the problem at hand becomes allocation of sufficient resources for each resolution. The analysis provides the critical number of AC resources to avoid this infeasibility.} In this way, the analysis should rather be used for low number of resources in these admission channel $M_{AC}$ \mgrevvv{where the assumption for deterministic arrivals holds}. 

For varying the number of resources in resolution channel $M_{RC}$, almost a linear behavior is observed for the rejection ratio. This is expected, since a better parallelization is enabled and resolutions with high number of users are almost linearly parallelizable \cite{gursu_mpcta}. 

The results for varying $M_{RC}$ in Fig.~\ref{fig:ana_sim_comp}: the increase in parallelization results in decreased rejection ratio as expected. For the really low $M_{RC}$ region, the curve has a better fit as explained with the analysis. After the \mgrevvv{Markovian} behavior for the number of servers vanishes, the curve deviates. Here we can emphasize a take out message for $p_r$. Varying the $M_{AC}$ has an expected behavior. With low $M_{AC}$, increasing the $M_{AC}$ exponentially decreases $p_r$ then after a certain number of resources it saturates to a linear decrease. This behavior is similar to a queue close to the stability limit. For $M_{RC}$ we have a linear decrease with a greater pace compared to the linear region of $M_{AC}$. Thus, we can conclude that a rule of thumb for dimensioning the resources for $M_{AC}$ and the $M_{RC}$ is: (1) allocate enough resource to $M_{AC}$ such that exponential $p_r$ behavior is overcome and (2) all the remaining resources are allocated to $M_{RC}$. The exponential region limit can be determined Eq.~\eqref{eq:pr} and taking the dip of the waterfall region as observed from Fig.~\ref{fig:ana_sim_comp}. For example in Fig.~\ref{fig:ana_sim_pbraq_mj_n45}, $M_{AC}=5$ and in Fig.~\ref{fig:ana_sim_pbraq_mj_n25} $M_{AC}=4$ should be selected and all other resources should be allocated to $M_{RC}$. 

\par Through the provided insights for the dimensioning of the algorithm, we now set the resources of AC/DC-RA accordingly and compare with the state of the art. 


\subsection{Comparison with Baseline}

We select the Dynamic Access Barring (DAB) algorithm as a baseline \cite{wang2015optimal}. This algorithm is an improved version of the access class barring algorithm currently used in LTE RACH. Through a backlog estimation the barring factor is updated dynamically. The barring of users enables optimal saturation throughput of Slotted ALOHA. It is also used with multiple QoS classes such that one class is prioritized over other, such that the no-priority class is fully barred when there are requests from the prioritized class. A dynamic barring factor is still applied to the prioritized class to guarantee optimal throughput. 

For AC/DC-RA we allocate 4 resources for each admission channels for each Class 1 and 2 and 12 resources for the resolution channel allocating 20 resources in total. For DAB algorithm we also allocate 20 resources to have a fair comparison. We use a deadline to refer to the delay constraint for comparison.

For AC/DC-RA we enable such prioritization through admission control, where one class is only accepted after the other class is fully admitted. Since we want to emphasize the priorities and the guarantee aspects, we use the same requirements for both classes. In order to show that the system can outperform the state of the art in extremely critical situations, we assume a Beta distributed arrival scenario representing bursty arrivals of M2M communications \cite{3rdGenerationPartnershipProject3GPP}. We have an activation time of $T_A=100$ slots for the beta arrival and we have other parameters of the distribution set as in the reference. \mgrev{There is an \textit{imbalance} between different traffic classes. The imbalance reflects a population ratio difference between traffic of two classes. We keep the naming as Class 1 and Class 2 where Class 1 denotes the prioritized class. However, an adjective is added to the classes to point out the traffic imbalance situation. These adjectives are Low and High, where the High class has 10 times more users than the Low class.}

\begin{figure*}[!htbp]
	\centering
	\begin{subfigure}{0.44\textwidth}
		\includegraphics[width=1\textwidth]{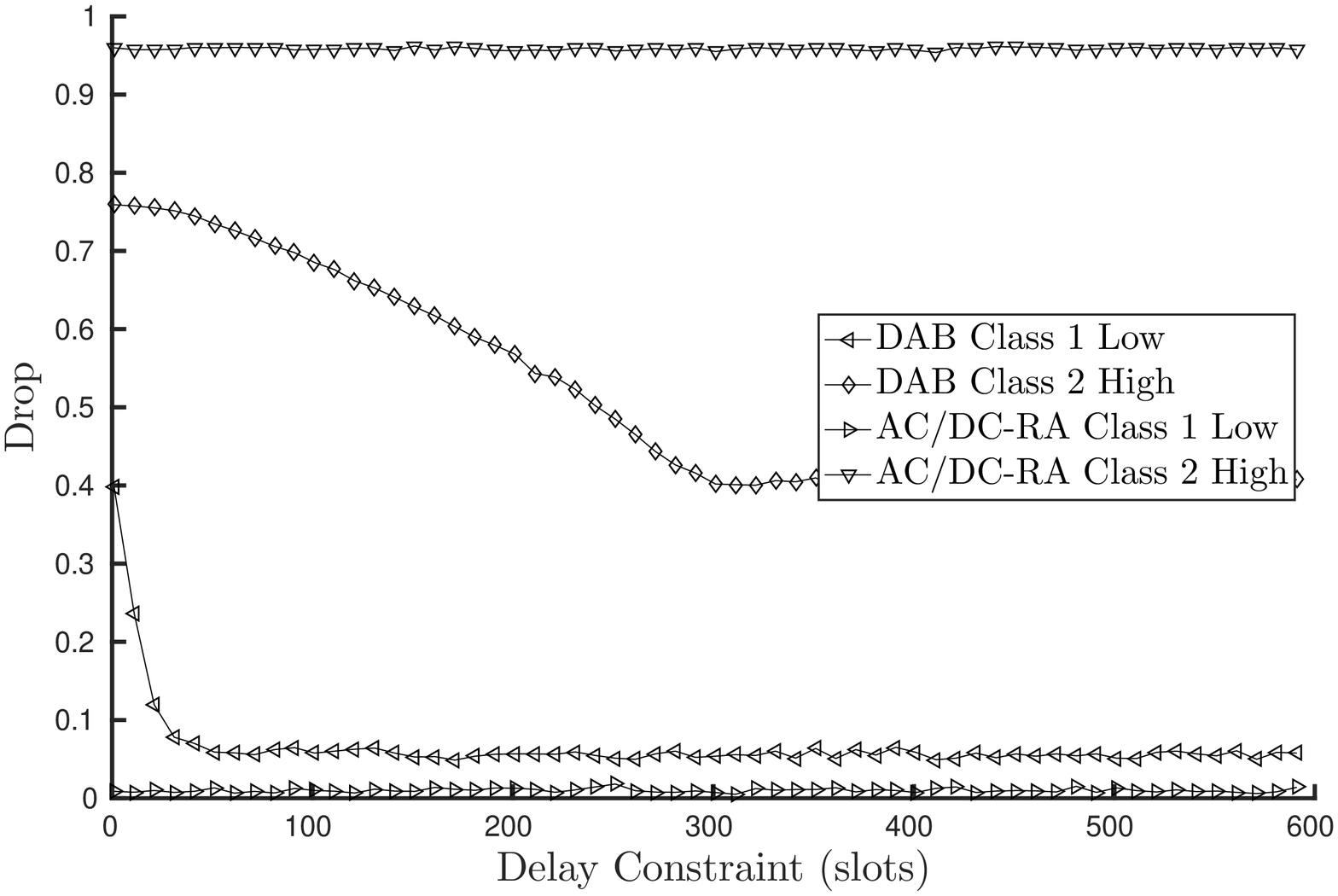}
		\caption{ Number of users for Class 2 $N_{high} = 2000$ and for Class 1 $N_{low} = 200$.}
		\label{fig:device_vary_comp14}
	\end{subfigure}	
	\begin{subfigure}{0.44\textwidth}
	\includegraphics[width=1\textwidth]{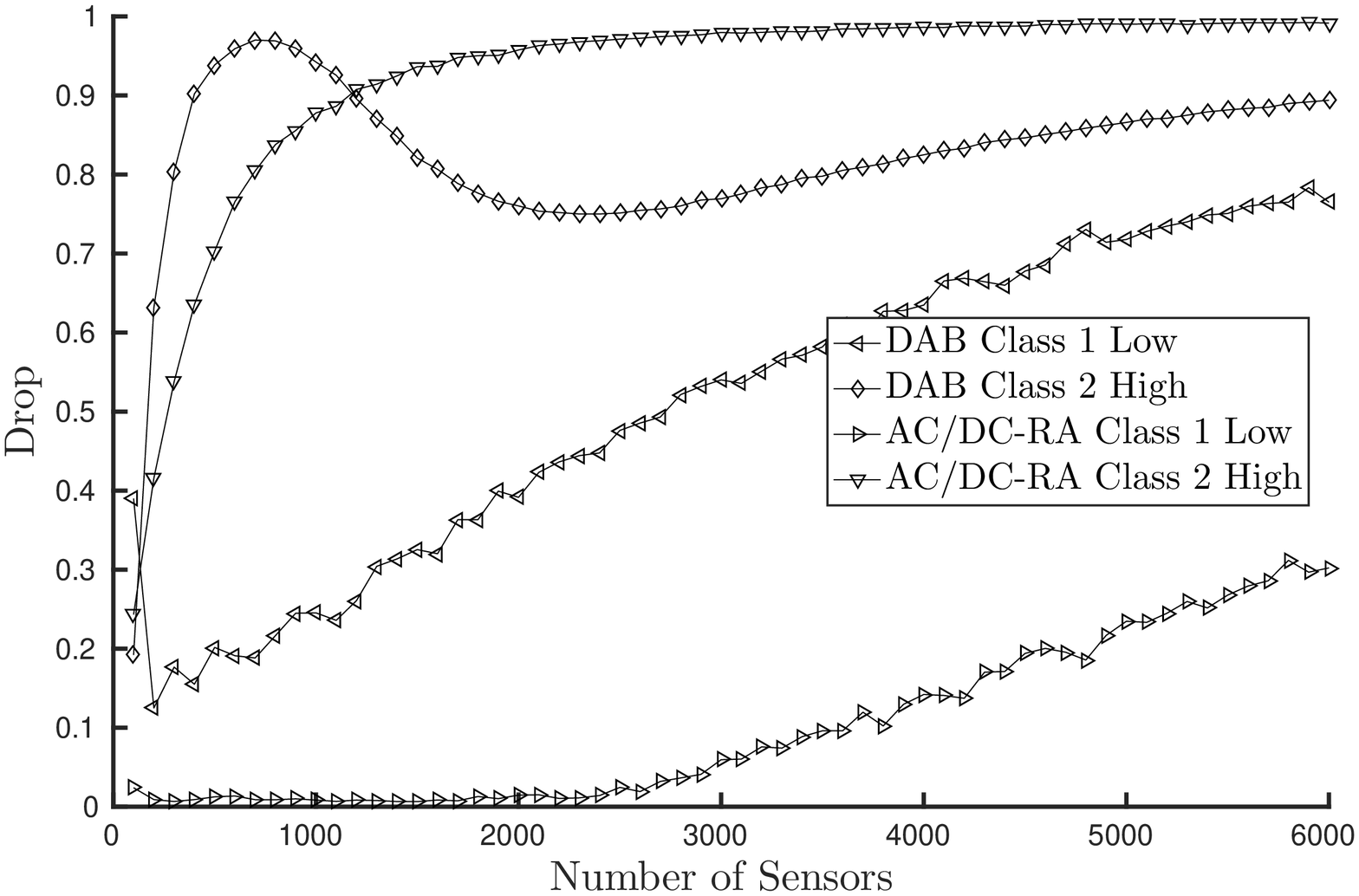}
	\caption{ Deadline $L = 10$ , Class 1 has 10 \% population of Class 2.}
	\label{fig:device_vary_comp16}
	\end{subfigure}
	\begin{subfigure}{0.44\textwidth}
		\includegraphics[width=1\textwidth]{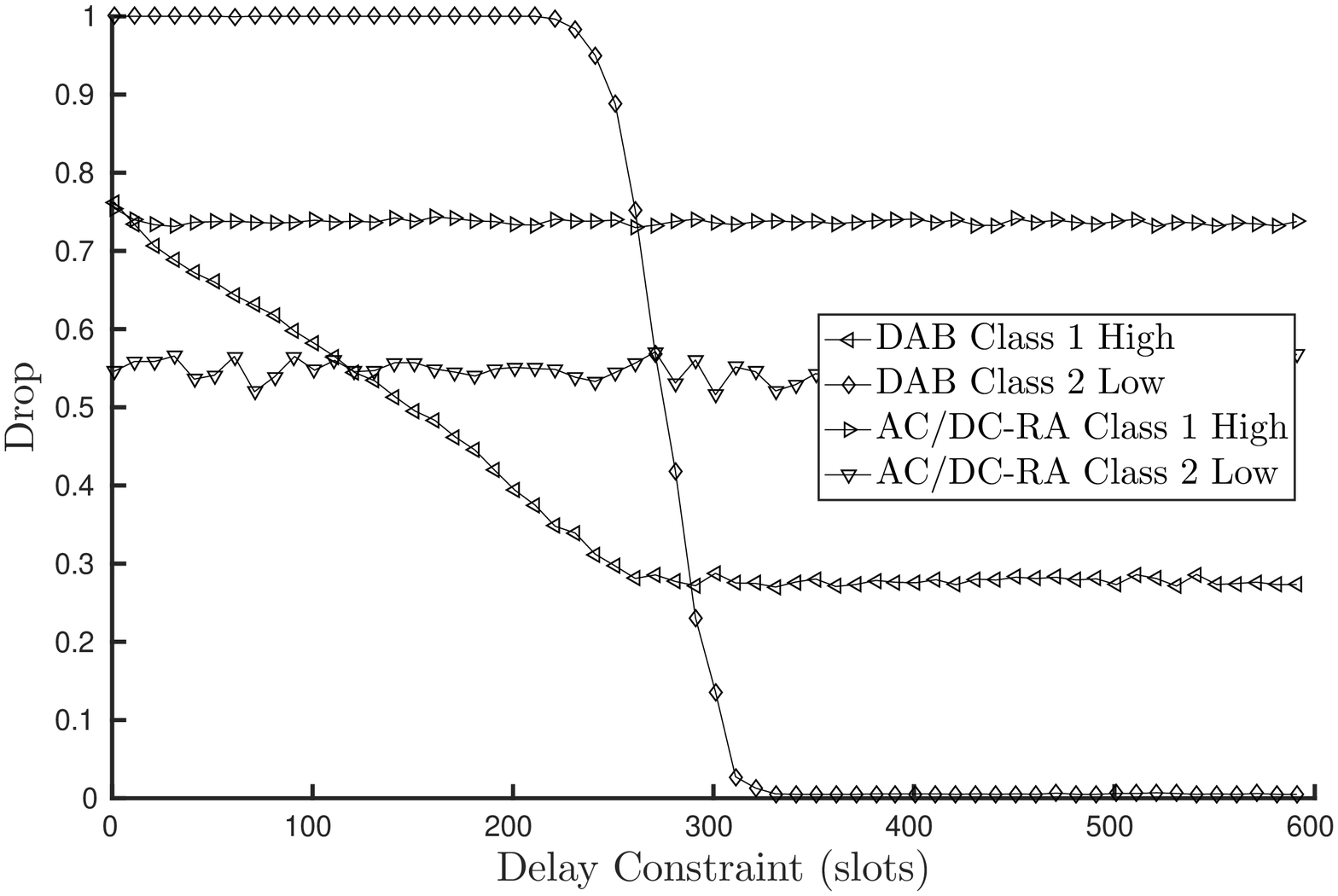}
		\caption{ Number of users for Class 1 $N_{high} = 2000$ and for Class 2 $N_{low} = 200$.}
		\label{fig:device_vary_comp15}
	\end{subfigure}
	\begin{subfigure}{0.44\textwidth}
		\includegraphics[width=1\textwidth]{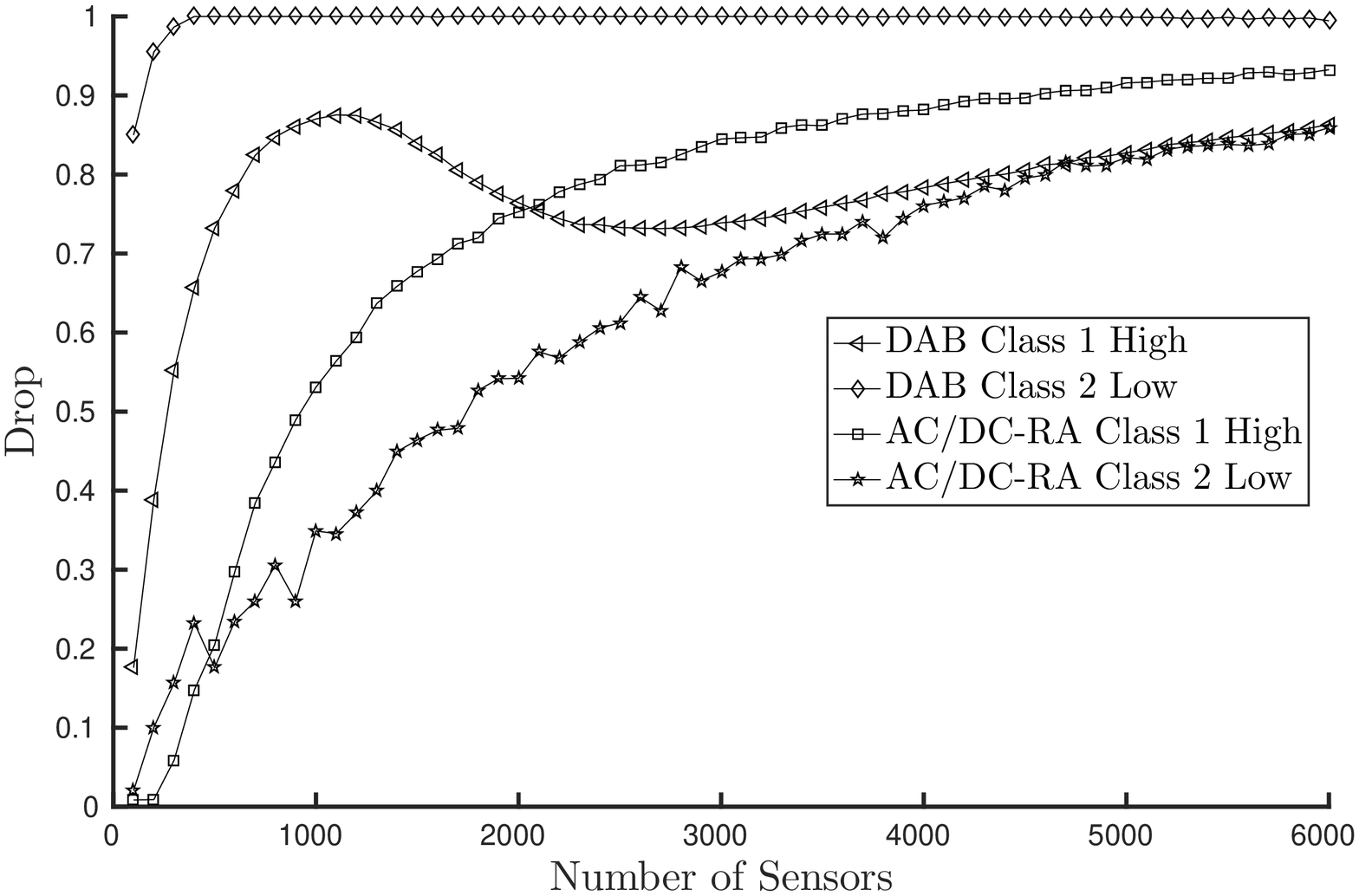}
		\caption{ Deadline $L = 10$ , Class 2 has 10 \% population of Class 1.}
		\label{fig:device_vary_comp17}
	\end{subfigure}
	\begin{subfigure}{0.44\textwidth}
	\includegraphics[width=1\textwidth]{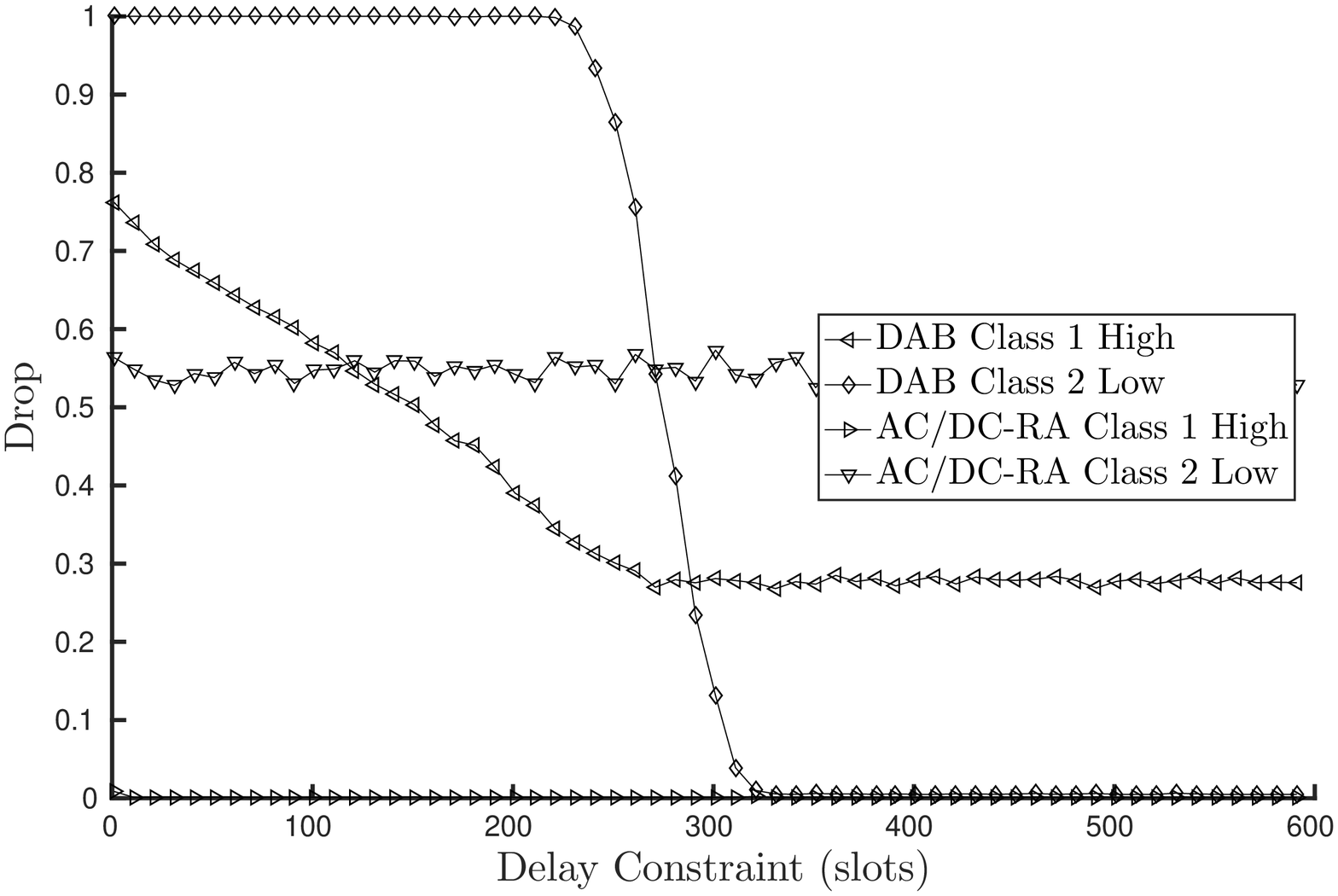}
	\caption{ Number of users for Class 1 $N_{high} = 2000$ and for Class 2 $N_{low} = 200$ with AC/DC-RA (buffering).}
	\label{fig:device_vary_buffer2}
	\end{subfigure}
	\begin{subfigure}{0.44\textwidth}
		\includegraphics[width=1\textwidth]{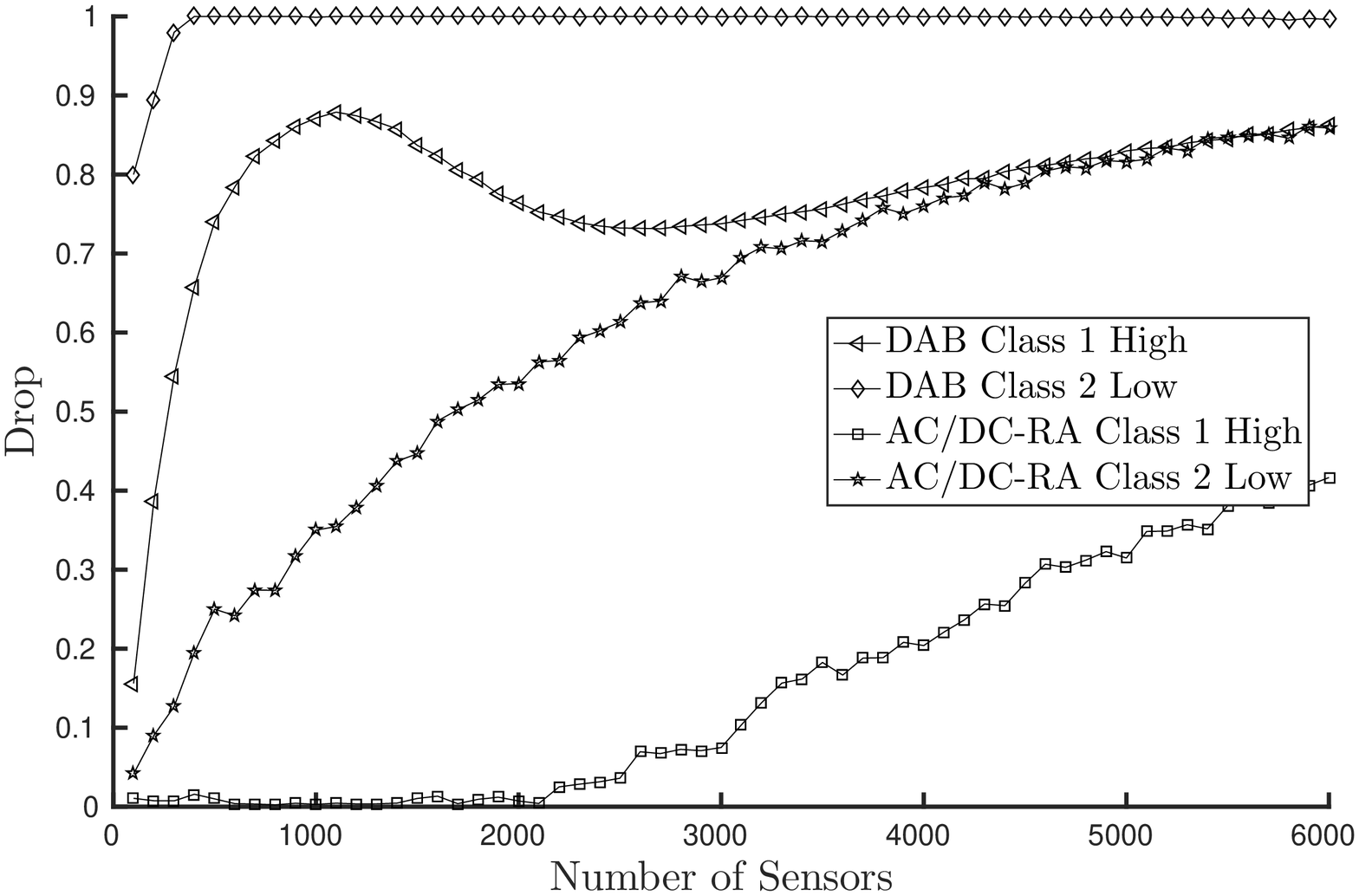}
		\caption{ Deadline $L = 10$ , Class 2 has 10 \% population of Class 1 with AC/DC-RA (buffering).}
		\label{fig:device_vary_buffer1}
	\end{subfigure}

	\caption{Comparison of AC/DC-RA with DAB varying the delay constraint $L$ and fixing the number of devices $N$. The traffic imbalance is introduced the one class has the number of users depicted as on the x-axis (denoted as High) while the other class has 10 \% of these users (denoted as Low).}
	\label{fig:device_vary_imba_comp}
\end{figure*}

\mgrev{In Fig.~\ref{fig:device_vary_imba_comp} we have plotted the AC/DC-RA against the baseline with the traffic imbalance. In Fig.~\ref{fig:device_vary_comp14} we have plotted the drop plus rejected ratio for varying delay constraints for Low Class 1 with 200 users and High Class 2 with 2000 users. The decrease in the number of users in the Low Class 1 results in an increased percentage of serviced users and most of the High Class 2 is blocked out. In Fig.~\ref{fig:device_vary_comp15} we have plotted the drop plus rejected ratio for varying delay constraints for Low Class 2 with 200 users and High Class 1 with 2000 users to represent a more scarce scenario. The Low Class 2 that uses DAB achieves a lower drop ratio thanks to low number of users. Some of the users from High Class 1 cannot be resolved in time even though the delay constraint is large. This is due to limited resource in RC that cannot react to burst arrivals. Interestingly for DAB, with larger delay constraints both classes achieve lower drop ratios compared to AC/DC-RA. This stems from the fact that obtaining delay and multiplicity information in the scenarios with relaxed delay constraints is not necessary for timely resolution. And the loss in resources to obtain this information cannot be made up with increased efficiency in the resolution channel. In Fig.~\ref{fig:device_vary_comp16} and Fig.~\ref{fig:device_vary_comp17} the drop plus rejected ratio is plotted against varying number of users. The x-axis depicts the number of users for High Class 2 in Fig.~\ref{fig:device_vary_comp16} and High Class 1 in Fig.~\ref{fig:device_vary_comp17} where the Low Class has 10\% of the population of the value depicted on x-axis for the High Class. On the other hand, such information is crucial for Low Class 1, as almost all users from that class can fit in the resolution channel. Thus, all multiplicity information obtained from AC is used. However, the information obtained for High Class 1 cannot help as there is not enough capacity in the resolution channel and these users have to be rejected irrespective of their multiplicity. This guides us to an important conclusion that if there is low amount of resolution channels and relaxed delay constraints the state of the art protocols can perform better. In Fig.~\ref{fig:device_vary_buffer1} and Fig.~\ref{fig:device_vary_buffer2} we have enabled admission of the users to a later available resource such that a certain waiting is enforced before accessing the resolution channel. We observe that this improves the performance but as the number of resources are limited all of the Class 2 users cannot be served.}

\section{Related Work}
\label{sec:related}






To the best of our knowledge this work is the first work that uses and admission control before the contention resolution. In the following, we discuss the admission control work that acts between contention resolution and scheduling as the most relevant state of the art.

\textbf{Admission control:} There are works that use admission control after random access for scheduling grants. A work from Bell Labs \cite{karol1995distributed} proposes a protocol that is called Distributed Queuing Request Update Multiple Access (DQRUMA). A scheduler keeps track of a distributed queue. The \textit{distributed queue} is the buffer status of multiple devices. When a device has a packet, it can place a request on the random access. This request can collide and the collision is resolved with a tree resolution algorithm. This is called the \textit{request} part of the algorithm. Through the state of the queues the scheduler decides whom to schedule. If a device is allocated a resource, then it can send a packet. At the end of a packet, a one bit header is added to notify that it has more packets. This approach is called piggy-backing and it \textit{update}s the distributed queue. In general, it has a similar structure as the LTE-A system in terms of access granting logic. Two particular differences are that tree resolution is used and through piggy-backing the load on the random access channel is decreased. A similar adaptation for the current mobile networks is also proposed in \cite{laya2016goodbye}. \mgrev{Distributed Queuing is an adaptation of the tree resolution protocol for requests. Thus, long waiting times for long data packets are avoided through transmission of request packets. In our work we measure efficiency in terms of slots such that use of requests or packets would not differentiate our results. But, it may make sense to use requests in certain scenarios so that the approach is expandable to use-cases where packet size matters. Here, we want to emphasize that the major difference of the MP-CTA is the \mgrevv{stochastic delay constraint} achieving capabilities due to parallelization of the tree resolution. Up to the best of our knowledge \mgrevv{stochastic delay constraint} access has been neglected by the DQ protocols. Some of the most recent work on DQ have taken a load reactivity direction. In \cite{bui2018design} the adaptation of the DQ protocol to LTE with load reactivity is provided and an analytical delay profile and transmission for group paging scenario using distributed queuing can be found in \cite{cheng2018modeling}. }  

 Admission control after random access is evaluated in \cite{abdrabou2008stochastic},\cite{song2007improving} for stochastic delay guarantees for calls over the IEEE 802.11 standard. There are also techniques which provide load adaptivity \cite{vilgelm2017latmapa} enabling optimal resource separation between various classes instead of using an admission control.


\textbf{Delay constrained random access:} The \mgrevv{stochastic delay constraint} in random access can be investigated mainly in two branches. First branch assumes that the \textit{arrival distribution} is known such that the resolution can provide guarantees for that arrival setting. In this branch, the total number of devices is assumed to be known. However, the exact activation time of each device is not known. Thus, contention algorithms are optimized with the knowledge of the total number of devices. In \cite{stefanovic2013joint}, authors suggests that devices are polled to the access channel. After each passing time-slot the probability of access decreases where after some time there are only idle channels. They also suggest that probability is modified exponentially. The set of outcomes is fed to a maximum likelihood estimator to provide the total number of backlogged devices. Through the knowledge obtained through polling they allocate required number of resources for contention. However, in case polling is done periodically it can translate into added delay. Another work with known number of users is \cite{stefanovic2017frameless} where authors have investigated resolution of certain number of users via successive interference cancellation capability within a certain limited amount of time. A recent work \cite{vilgelm2018icc} uses stochastic network calculus to provide stochastic bounds on the delay for dynamic access barring. This would make DAB also usable under an admission controlled way like we provide here. 

Second branch is where arrivals shaped with collision avoidance techniques or polling to arrange in a manner that the resolution can provide guarantees. In case a sensor can access the random access channel as soon as possible, it would save two critical resources time and energy. The state of the art is mostly dealing with this assuming that the central station can \textbf{detect} the number of active devices on each channel. For instance a recent work \cite{polling_abbas} suggested to modify the random access behavior with successive interference cancellation technique assuming that the base station can detect the multiplicity of the number of active devices in order to guarantee certain service requirements. However, in case these assumptions are not valid, there would be no guarantees. And if the behavior of a device is constrained to obey these assumptions, the random access channel cannot be used \textit{randomly} anymore. Another work  \cite{gursu2017hybrid} has shaped arrivals through pre-backoff with the size knowledge of the burst arrival, where an optimized tree resolution performance is achieved thanks to the collision avoidance. 
\section{Conclusion}
\label{sec:conc}
In this paper we introduce a new random access protocol AC/DC-RA - Admission Control based Traffic-Agnostic Delay-Constrained Random Access. This protocol changes the random access paradigm with an addition of an \textit{admission control decision}. The admission control decision is based on a novel collision size estimation and analytical modeling of the resolution. This estimation enables an accurate guess for the delay of a contention resolution. 

We furthermore provide a Markov Chain based analysis to investigate the behavior of the protocol. Then we show that the dimensioning problem stemming from this protocol can be modeled in closed form and solved offline. We then compare the algorithm to a state of the art approach and show that in order to guarantee a resolution the proposed modifications are necessary. Otherwise guaranteeing the reliability is only best effort. We claim that such a paradigm shift is necessary to use an admission channel to enable stability and scalability of random access against any type of unexpected traffic, as it would be the case for M2M communications.

\mgrevvv{Future work can investigate the effects of adjusting the number of channels in a dynamic fashion, e.g., each slot. This can unravel the overhead of broadcasting the system updates. As the periodicity may depend on the duty cycle of sensors a large overhead may be required.}

  \appendices
\section{Proof for expectation calculation of Coupon Collector's Problem with unequal probabilities}
\label{app:proof1}

We start by repeating the probability $p_i$ that any user accesses a channel $i$. If we have $z$ users in the system, we have a mean arrival of $\lambda_i = p_i\cdot z$ on the $i^{\text{th}}$ resource. Thus the idle probability on that resource is $	 e^{-\lambda_i} = e^{-p_i\cdot z}$. Non idle probability on that resource is $1-e^{-p_i\cdot z}$. If we multiply this probability for all resource that had a busy signal we get,$ \prod\limits_{i \in \mathcal{M}_{s+c}}(1-e^{-p_i z}), $
which is the probability to have non-idle on all the busy resource. This probability can be used for the likelihood \mgrevvv{of having all busy signals for the set of $\mathcal{M}_{s+c}$ resources.}.  However, if we take the probability of observing at least one idle in the busy resources $
1-\prod\limits_{i \in \mathcal{M}_{s+c}}(1-e^{-p_i z})$, 
\mgrevvv{ with increasing $z$, this probability goes to zero and taking the expectation \textit{for each added user} where the sum goes to infinity is no problem}. The expectation gives then the expected value of users to be added until no idle in the selected resources $\mathcal{M}_{s+c}$ is observed,

\begin{equation}
\Exp{Z | \mathcal{M}_{s+c}} = \sum\limits_{z=0}^{\infty} \left( 1 - \prod\limits_{i \in \mathcal{M}_{s+c}}(1-e^{-p_i z}) \right).
\end{equation}

\section{Proof of Theorem \ref{thm:1}}
\label{sec:app3}

Collision probability $p_c$  is a sub problem of probability of observing $u$ balls in any bins, with $N$ balls into $M$ bins with unequal probabilities. We start this with re-defining the set of possible bins as $\mathcal{M}  = \{ 1, ..., M \}$. Then we define $\mathbf{S}^J$ that denotes the sequence for all possible $J$-ary combination sequences of the elements of set $\mathcal{M}$. An example would be $\mathbb{S}^2 = ( \{1,2\},\{1,3\},\{2,3\} )$,
with $\mathcal{M}  = \{ 1, 2, 3 \}$ and $J=2$. We will also use the term $\mathbb{S}_{x,y}^J$ where $x \in  (1, ... , {M \choose J} )$ denotes different sets in the sequence and $y \in ( 1, ..., J  )$ denotes different elements of each combination sequence. \mgrevv{From the} example we have $\mathbb{S}^2_{1,2} = 2$ and $\mathbb{S}^2_{3,2} = 3$.  We will also use $\mathbb{S}_{x}^J$ when we want to refer just to the set.
\par Now we denote $W^L_{x,y}(u)$ as the probability function for  selecting $u$ users out of $N$ users for the $L^{\text{th}}$ time  as given in $W^L_{x,y}(u) = {{N - u(y-1)}\choose{u}} \left( p_{{\mathbb{S}_{x,y}^L}} \right)^{u}$.
Now we denote $Z^J_{x,y}(u)$ as the probability function for  $N-J\cdot u$ users out of $N$ users selecting all the other frequencies except the onces denoted by the set $x$ as given in $Z^J_{x,y}(u) = \left( 1 - \prod_{z=1}^{y} p_{{\mathbb{S}_{x,z}^J}} \right)^{N-J\cdot u}$. Using these we define the probability function to obtain $J$ occurrence of $u$ users out of $N$ users with a recursive calculation 
\begin{equation}
P_{J}^x (u)  = \left(\prod_{y=1}^{J} W^{J}_{x,y}(u) \right) Z^J_{x,y}(u)   - \left( \sum_{j=J+1}^{max(J)} \sum_{ x \in \{ \mathbb{S}_x^J \subset \mathbb{S}_x^j \} } P_{j}^x (u) \right)
\label{eqn:iterative_p}
\end{equation}
where $max(J)$ is maximum number of occurrence of $u$ given $N$ users which is given with $ min(\lfloor \frac{N}{u}\rfloor,M)$. In Eq.~\eqref{eqn:iterative_p} the upper part calculates the joint probability of having $J$ occurrence of $u$ users, while the lower part is subtracting the probabilities for $j>J$ occurrences. After we have non-overlapping probabilities for all occurrences of $u$, i.e., probability to have \textbf{just} $J$ occurrence of $u$ users, we can sum them up to have the probability to obtain $u$ users,
\begin{equation}
p_c (u) = \sum_{j=1}^{J} \sum_{x \in \mathbb{S}^J_x} \left(  P_{J}^x (u) \cdot l \right)
\label{eq:coll}
\end{equation}
where we multiply with the occurrence of $u$ users since we treat each outcome independently and we have to weigh accordingly.

\bibliographystyle{IEEEtran}
\bibliography{decara}

\end{document}